\documentclass[11pt]{article}
\usepackage{amsmath,amssymb,amscd,amsthm}
\usepackage{epsfig}

\newtheorem{theorem}{Theorem}
\newtheorem{lemma}{Lemma}
\newtheorem{proposition}{Proposition}
\newtheorem{remark}{Remark}

\newtheorem{corollary}{Corollary}

\newtheorem{claim}{Claim}

\newcommand{\f}[2]{\frac{#1}{#2}}
\newcommand{\dpr}[2]{\langle #1,#2 \rangle}

\newcommand{\de}{\delta}

\newcommand{\ve}{\varepsilon}

\newcommand{\la}{\lambda}

\newcommand{\si}{\sigma}

\newcommand{\om}{\omega}

\newcommand{\suml}{\sum\limits}

\newcommand{\beq}{\begin{equation}}
\newcommand{\eeq}{\end{equation}}
\newcommand{\beqna}{\begin{eqnarray*}}
\newcommand{\eeqna}{\end{eqnarray*}}
\newcommand{\beqn}{\begin{equation*}}
\newcommand{\eeqn}{\end{equation*}}
\newcommand{\bp}{\begin{proof}}
\newcommand{\ep}{\end{proof}}
\newcommand{\bprop}{\begin{proposition}}
\newcommand{\eprop}{\end{proposition}}
\newcommand{\bt}{\begin{theorem}}
\newcommand{\et}{\end{theorem}}
\newcommand{\bex}{\begin{Example}}
\newcommand{\eex}{\end{Example}}
\newcommand{\bc}{\begin{corollary}}
\newcommand{\ec}{\end{corollary}}
\newcommand{\bcl}{\begin{claim}}
\newcommand{\ecl}{\end{claim}}
\newcommand{\bl}{\begin{lemma}}
\newcommand{\el}{\end{lemma}}

\newenvironment{proof1}%
{\begin{trivlist} \item[]{\em Proof }}%
{\hspace*{\fill}$\rule{.3\baselineskip}{.35\baselineskip}$\end{trivlist}}

\begin{document}

\title{\bf Asymptotic stability of small solitons in the discrete
nonlinear Schr\"{o}dinger equation in one dimension}

\author{P.G. Kevrekidis$^1$, D.E. Pelinovsky$^2$, and A.
Stefanov$^3$ \\
{\small $^{1}$ Department of Mathematics and Statistics, University
of Massachusetts, Amherst, MA 01003}\\
{\small $^{2}$ Department of Mathematics and Statistics, McMaster
University,
Hamilton, Ontario, Canada, L8S 4K1} \\
{\small $^{3}$ Department of Mathematics, University of Kansas, 1460
Jayhawk Blvd, Lawrence, KS 66045--7523} }

\maketitle

\begin{abstract}
Asymptotic stability of small solitons in one dimension is proved
in the framework of a discrete nonlinear Schr\"{o}dinger equation with septic and
higher power-law nonlinearities and an external potential supporting a
simple isolated eigenvalue. The analysis relies on the dispersive
decay estimates from Pelinovsky \& Stefanov (2008) and the
arguments of Mizumachi (2008) for a continuous nonlinear
Schr\"{o}dinger equation in one dimension. Numerical simulations suggest that the
actual decay rate of perturbations near the asymptotically stable
solitons is higher than the one used in the analysis.
\end{abstract}

\section{Introduction}

Asymptotic stability of solitary waves in the context of
continuous nonlinear Schr\"{o}dinger equations in one, two, and
three spatial dimensions was considered in a number of recent
works (see Cuccagna \cite{cuccagna} for a review of literature). 
Little is known, however, about asymptotic
stability of solitary waves in the context of discrete nonlinear
Schr\"{o}dinger (DNLS) equations.

Orbital stability of a global energy minimizer under a fixed
mass constraint was proved by Weinstein \cite{weinstein} for the
DNLS equation with power nonlinearity
$$
i \dot{u}_n + \Delta_d u_n + |u_n|^{2 p} u_n = 0, \quad n \in
\mathbb{Z}^d,
$$
where $\Delta_d$ is a discrete Laplacian in $d$ dimensions and $p
> 0$. For $p < \frac{2}{d}$ (subcritical case), it is
proved that the ground state of an arbitrary energy exists,
whereas for $p \geq \frac{2}{d}$ (critical and supercritical
cases), there is an energy threshold, below which the ground state
does not exist.

Ground states of the DNLS equation with power-law nonlinearity
correspond to single-humped solitons, which are excited in
numerical and physical experiments by a single-site initial data
with sufficiently large amplitude \cite{KEDS}. Such experiments
have been physically realized in optical settings with both
focusing \cite{mora} and defocusing \cite{rosberg} nonlinearities.
We would like to consider long-time dynamics of the ground states
and prove their asymptotic stability under some assumptions on the
spectrum of the linearized DNLS equation. From
the beginning, we would like to work in the space of one spatial
dimension $(d = 1)$ and to add an external potential $V$ to the
DNLS equation. These specifications are motivated by physical
applications (see, e.g., the recent work of \cite{kroli} and
references therein for a relevant discussion). We hence write the
main model in the form
\begin{equation}
\label{dNLS} i \dot{u}_n = (-\Delta + V_n) u_n + \gamma |u_n|^{2p} u_n,
\quad n \in \mathbb{Z},
\end{equation}
where $\Delta u_n := u_{n+1} - 2 u_n + u_{n-1}$ and $\gamma = 1$
($\gamma = -1$) for defocusing (focusing) nonlinearity. Besides
physical applications, the role of potential $V$ in our work can
be explained by looking at the differences between the recent
works of Mizumachi \cite{Miz} and Cuccagna \cite{Cuc} for a
continuous nonlinear Schr\"{o}dinger equation in one dimension.
Using an external potential, Mizumachi proved asymptotic stability
of small solitons bifurcating from the ground state of the
Schrodinger operator $H_0 = -\partial_x^2 + V$ under some
assumptions on the spectrum of $H_0$. He needed only spectral
theory of the self-adjoint operator $H_0$ in $L^2$ since spectral
projections and small nonlinear terms were controlled in the
corresponding norm. Pioneering works along the same lines are
attributed to Soffer--Weinstein \cite{SW1,SW2,SW3}, Pillet \&
Wayne \cite{PW}, and Yao \& Tsai \cite{YT1,YT2,YT3}. Compared to
this approach, Cuccagna proved asymptotic stability of nonlinear
space-symmetric ground states in energy space of the continuous
nonlinear Schr\"{o}dinger equation with $V \equiv 0$. He had to
invoke the spectral theory of non-self-adjoint operators arising
in the linearization of the nonlinear Schr\"{o}dinger equation at
the ground state, following earlier works of Buslaev \& Perelman
\cite{BP1,BP2}, Buslaev \& Sulem \cite{BS}, and Gang \& Sigal
\cite{GS1,GS2}.

Since our work is novel in the context of the DNLS equation, we
would like to simplify the spectral formalism and to focus on
nonlinear analysis of asymptotic stability. This is the main
reason why we work with small solitons bifurcating from the ground
state of the discrete Schrodinger operator $H = -\Delta + V$. We
will make use of the dispersive decay estimates obtained recently
for operator $H$ by Stefanov \& Kevrekidis \cite{SK} (for $V
\equiv 0$), Komech, Kopylova \& Kunze \cite{KKK} (for compact
$V$), and Pelinovsky \& Stefanov \cite{PS} (for decaying $V$).
With more efforts and more elaborate analysis, our results can be
generalized to large solitons with or without potential $V$ under
some restrictions on spectrum of the non-self-adjoint operator
associated with linearization at the nonlinear ground state.

From a technical point of view, many previous works on asymptotic
stability of solitary waves in continuous nonlinear Schr\"{o}dinger equations 
address critical and supercritical cases, which in $d = 1$
corresponds to $p \geq 2$. Because the dispersive decay in
$l^1-l^{\infty}$ norm is slower for the DNLS equation, the
critical power appears at $p = 3$ and the proof of
asymptotic stability of discrete solitons can be developed for $p \geq 3$. The most
interesting case of the cubic DNLS equation for $p = 1$ is
excluded from our consideration. To prove asymptotic stability of
discrete solitons for $p \geq 3$, we extend the pointwise
dispersive decay estimates from \cite{PS} to Strichartz estimates,
which allow us for a better control of the dispersive parts of the
solution. The nonlinear analysis follows the steps in the proof of
asymptotic stability of continuous solitons by Mizumachi
\cite{Miz}.

In addition to analytical results, we also approximate time
evolution of small solitons numerically in the DNLS equation
(\ref{dNLS}) with $p = 1,2,3$. Not only we confirm the asymptotic
stability of discrete solitons in all the cases but also we find
that the actual decay rate of perturbations near the small soliton
is faster than the one used in our analytical arguments.

The article is organized as follows. The main result for $p \geq 3$ is
formulated in Section 2. Linear estimates are derived in Section
3. The proof of the main theorem is developed in Section 4.
Numerical illustrations for $p = 1, 2, 3$ are discussed in Section
5. Appendix A gives proofs of technical formulas used in Section
3.

{\bf Acknowledgement.} When the paper was essentially complete, we
became aware of a similar work of Cuccagna \& Tarulli \cite{CT},
where asymptotic stability of small discrete solitons of the DNLS
equation (\ref{dNLS}) was proved for $p \geq 3$. 

Stefanov's research is supported in part by NSF-DMS 0701802. 
Kevrekidis' research is supported in part by NSF-DMS-0806762, 
NSF-CAREER and the Alexander von Humboldt Foundation.

\section{Preliminaries and the main result}

In what follows, we use bold-faced notations for vectors in
discrete spaces $l_s^1$ and $l_s^2$ on $\mathbb{Z}$ defined by
their norms
$$
\| {\bf u} \|_{l^1_s} := \sum_{n \in \mathbb{Z}} (1+n^2)^{s/2}
|u_n|, \quad \| {\bf u} \|_{l^2_s} := \left( \sum_{n \in
\mathbb{Z}} (1+n^2)^{s} |u_n|^2 \right)^{1/2}.
$$
Components of ${\bf u}$ are denoted by regular font, e.g. $u_n$
for $n \in \mathbb{Z}$.

We shall make the following assumptions on the external potential
${\bf V}$ defined on the lattice $\mathbb{Z}$ and on the spectrum
of the self-adjoint operator $H = -\Delta + {\bf V}$ in $l^2$.

\begin{itemize}
\item[(V1)] ${\bf V} \in l^1_{2\sigma}$ for a fixed $\sigma >
\frac{5}{2}$.

\item[(V2)] ${\bf V}$ is generic in the sense that no solution
$\mbox{\boldmath $\psi$}_0$ of equation $H \mbox{\boldmath
$\psi$}_0 = 0$ exists in $l^2_{-\sigma}$ for $\frac{1}{2} < \sigma
\leq \frac{3}{2}$.

\item[(V3)] ${\bf V}$ supports exactly one negative eigenvalue
$\omega_0 < 0$ of $H$ with an eigenvector $\mbox{\boldmath
$\psi$}_0 \in l^2$ and no eigenvalues above $4$.
\end{itemize}

The first two assumptions (V1) and (V2) are needed for 
the dispersive decay estimates developed in \cite{PS}. The last
assumption (V3) is needed for existence of a family $\mbox{\boldmath
$\phi$}(\omega)$ of real-valued decaying solutions of
the stationary DNLS equation
\begin{equation}
\label{stationaryDNLS} (-\Delta + V_n) \phi_n(\omega) + \gamma
\phi_n^{2p+1}(\omega) = \omega \phi_n(\omega), \quad n \in
\mathbb{Z},
\end{equation}
near $\omega = \omega_0 < 0$. This is a standard local bifurcation
of decaying solutions in a system of infinitely many algebraic equations
(see \cite{Nirenberg} for details).

\begin{lemma}[Local bifurcation of stationary solutions]
\label{lemma-bifurcation} Assume that ${\bf V} \in l^{\infty}$ and
that $H$ has an eigenvalue $\omega_0$ with a normalized
eigenvector $\mbox{\boldmath $\psi$}_0 \in l^2$ such that $\|
\mbox{\boldmath $\psi$}_0 \|_{l^2} = 1$. Let $\epsilon := \omega -
\omega_0$, $\gamma = +1$, and $\epsilon_0 > 0$ be sufficiently
small. For any $\epsilon \in (0,\epsilon_0)$, there exists an
$\epsilon$-independent constant $C > 0$ such that the stationary
DNLS equation (\ref{stationaryDNLS}) admits a solution
$\mbox{\boldmath $\phi$}(\omega) \in
C^2([\omega_0,\omega_0+\epsilon_0],l^2)$ satisfying
$$
\left\| \mbox{\boldmath $\phi$}(\omega) -
\frac{\epsilon^{\frac{1}{2p}} \mbox{\boldmath $\psi$}_0}{\| \mbox{\boldmath $\psi$}_0
\|^{1+\frac{1}{p}}_{l^{2p+2}}}
\right\|_{l^2} \leq C \epsilon^{1 + \frac{1}{2p}}.
$$
Moreover, the solution $\mbox{\boldmath $\phi$}(\omega)$ decays
exponentially to zero as $|n| \to \infty$.
\end{lemma}

\begin{remark}
\label{remark-bifurcation} Because of the exponential decay of
$\mbox{\boldmath $\phi$}(\omega)$ as $|n| \to \infty$, the
solution $\mbox{\boldmath $\phi$}(\omega)$ exists in $l^2_{\si}$
for all $\si \geq 0$. In addition, since $ \| \mbox{\boldmath
$\phi$}\|_{l^1} \leq C_{\si} \| \mbox{\boldmath $\phi$}
\|_{l^2_{\si}}, $ for any $\si > \frac{1}{2}$, the solution
$\mbox{\boldmath $\phi$}(\omega)$ also exists in $l^1$.
\end{remark}

\begin{remark}
The case $\gamma = -1$ with the local bifurcation to the domain
$\omega < \omega_0$ is absolutely analogous. For simplification,
we shall develop analysis for $\gamma = +1$ only.
\end{remark}

To work with solutions of the DNLS equation (\ref{dNLS}) for all
$t \in {\mathbb R}_+$ starting with some initial data at $t = 0$,
we need global well-posedness of the Cauchy problem for
(\ref{dNLS}). Because $H$ is a bounded operator from $l^2$ to
$l^2$, global well-posedness for (\ref{dNLS}) follows from simple
arguments based on the flux conservation equation
\begin{equation}
\label{balance} i \frac{d}{dt} |u_n|^2 = u_n (\bar{u}_{n+1} +
\bar{u}_{n-1}) - \bar{u}_n (u_{n+1}+u_{n-1})
\end{equation}
and the contraction mapping arguments (see \cite{PP} for details).

\begin{lemma}[Global well-posedness]
\label{lemma-wellposedness} Fix $\si \geq 0$. For any ${\bf u}_0
\in l^2_{\si}$, there exists a unique solution ${\bf
u}(t) \in C^1(\mathbb{R}_+,l^2_{\si})$ such that ${\bf
u}(0) = {\bf u}_0$ and ${\bf u}(t)$ depends continuously on
${\bf u}_0$.
\end{lemma}

\begin{remark}
Global well-posedness holds also on $\mathbb{R}_-$ (and thus on
$\mathbb{R}$) since the DNLS equation (\ref{dNLS}) is a reversible
dynamical system. We shall work in the positive time intervals
only.
\end{remark}

Equipped with the results above, we decompose a solution to the DNLS equation
(\ref{dNLS}) into a family of stationary solutions with time varying
parameters and a radiation part using the substitution
\begin{equation}
\label{decomposition}
{\bf u}(t) = e^{-i \theta(t)} \left( \mbox{\boldmath
$\phi$}(\omega(t)) + {\bf z}(t) \right),
\end{equation}
where $(\omega,\theta) \in \mathbb{R}^2$ represents a two-dimensional
orbit of stationary solutions ${\bf u}(t) = e^{-i\theta -i \omega t} \mbox{\boldmath
$\phi$}(\omega)$ (their time
evolution will be specified later) and ${\bf z}(t) \in
C^1(\mathbb{R}_+,l^2_{\sigma})$ solves the
time-evolution equation in the form
\begin{eqnarray}
\label{time-evolution-z} i \dot{{\bf z}} = (H-\omega) {\bf z} -
(\dot{\theta} - \omega) (\mbox{\boldmath $\phi$}(\omega) + {\bf z})
- i \dot{\omega}
\partial_{\omega} \mbox{\boldmath $\phi$}(\omega) +
{\bf N}(\mbox{\boldmath $\phi$}(\omega)+{\bf z}) - {\bf N}(\mbox{\boldmath $\phi$}(\omega)),
\end{eqnarray}
where $H = -\Delta + {\bf V}$, $[{\bf N}(\mbox{\boldmath
$\psi$})]_n = \gamma |\psi_n|^{2p} \psi_n$, and $\partial_{\omega}
\mbox{\boldmath $\phi$}(\omega)$ exists thanks to Lemma
\ref{lemma-bifurcation}. The linearized time evolution at the
stationary solution $ \mbox{\boldmath $\phi$}(\omega)$ involves
operators
$$
L_- = H - \omega + {\bf W}, \quad L_+ = H - \omega + (2p+1) {\bf
W},
$$
where $W_n = \gamma \phi_n^{2p}(\omega)$ and ${\bf W}$ decays
exponentially as $|n| \to \infty$ thanks to Lemma
\ref{lemma-bifurcation}. The linearized time evolution in
variables ${\bf v} = {\rm Re}({\bf z})$ and ${\bf w} = {\rm
Im}({\bf z})$ involves a symplectic structure which can be
characterized by the non-self-adjoint eigenvalue problem
\begin{equation}
\label{linearizedNLS} L_+ {\bf v} = - \lambda {\bf w}, \quad L_-
{\bf w} = \lambda {\bf v}.
\end{equation}
Using Lemma \ref{lemma-bifurcation}, we derive the following
result.

\begin{lemma}[Double null subspace]
For any $\epsilon \in (0,\epsilon_0)$, the linearized eigenvalue
problem (\ref{linearizedNLS}) admits a double zero eigenvalue with
a one-dimensional kernel, isolated from the rest of the spectrum.
The generalized kernel is spanned by vectors $({\bf
0},\mbox{\boldmath $\phi$}(\omega)), (- \partial_{\omega}
\mbox{\boldmath $\phi$}(\omega),{\bf 0}) \in l^2$ satisfying
$$
L_- \mbox{\boldmath $\phi$}(\omega) = {\bf 0}, \qquad L_+
\partial_{\omega} \mbox{\boldmath $\phi$}(\omega) = \mbox{\boldmath
$\phi$}(\omega).
$$
If $({\bf v},{\bf w}) \in l^2$ is symplectically orthogonal to the
double subspace of the generalized kernel, then
$$
\langle {\bf v},\mbox{\boldmath $\phi$}(\omega) \rangle = 0, \quad
\langle {\bf w},\partial_{\omega} \mbox{\boldmath $\phi$}(\omega)
\rangle = 0,
$$
where $\langle {\bf u},{\bf v} \rangle := \sum_{n \in \mathbb{Z}}
u_n \bar{w}_n$.
\end{lemma}

\begin{proof}
By Lemma 1 in \cite{PS}, operator $H$ has the essential spectrum
on $[0,4]$. Because of the exponential decay of ${\bf W}$ as $|n|
\to \infty$, the essential spectrum of $L_+$ and $L_-$ is shifted
by $-\omega \approx -\omega_0 > 0$, so that the zero point in the
spectrum of the linearized eigenvalue problem
(\ref{linearizedNLS}) is isolated from the continuous spectrum and
other isolated eigenvalues. The geometric kernel of the linearized
operator $L = {\rm diag}(L_+,L_-)$ is one-dimensional for
$\epsilon \in (0,\epsilon_0)$ since $L_- \mbox{\boldmath
$\phi$}(\omega) = {\bf 0}$ is nothing but the stationary DNLS
equation (\ref{stationaryDNLS}) whereas $L_+$ has an empty kernel
thanks to the perturbation theory and Lemma
\ref{lemma-bifurcation}. Indeed, for a small $\epsilon \in
(0,\epsilon_0)$, we have
$$
\langle \mbox{\boldmath $\psi$}_0, L_+ \mbox{\boldmath $\psi$}_0 \rangle =
2p \gamma \epsilon + {\cal O}(\epsilon^2) \neq 0.
$$
By the perturbation theory, a simple zero eigenvalue of $L_+$ for
$\epsilon = 0$ becomes a positive eigenvalue for $\epsilon > 0$
(if $\gamma = +1$). The second (generalized) eigenvector $(-
\partial_{\omega} \mbox{\boldmath $\phi$}(\omega),{\bf 0})$ is
found by direct computation thanks to Lemma
\ref{lemma-bifurcation}. It remains to show that the third
(generalized) eigenvector does not exist. If it does, it would
satisfy the equation
$$
L_- {\bf w}_0 = -\partial_{\omega} \mbox{\boldmath $\phi$}(\omega).
$$
However,
$$
\langle \mbox{\boldmath $\phi$}(\omega),\partial_{\omega}
\mbox{\boldmath $\phi$}(\omega) \rangle = \frac{1}{2} \frac{d}{d
\omega} \| \mbox{\boldmath $\phi$}(\omega) \|^2_{l^2}
= \frac{\epsilon^{\frac{1}{p}-1}}{2p \| \mbox{\boldmath $\psi$}_0\|^{2 + \frac{2}{p}}_{l^{2p+2}}}
\left( 1 + {\cal O}(\epsilon) \right) \neq 0
$$
for $\epsilon \in (0,\epsilon_0)$ by Lemma
\ref{lemma-bifurcation}. Therefore, no ${\bf w}_0 \in l^2$ exists.
\end{proof}

To determine the time evolution of varying parameters $(\omega,\theta)$
in the evolution equation (\ref{time-evolution-z}), we shall
add the condition that ${\bf z}(t)$ is symplectically orthogonal
to the two-dimensional null subspace of the linearized problem
(\ref{linearizedNLS}). To normalize the eigenvectors uniquely, we set
\begin{equation}
\label{eigenvectors-normalized} \mbox{\boldmath $\psi$}_1 =
\frac{\mbox{\boldmath $\phi$}(\omega)}{\|\mbox{\boldmath
$\phi$}(\omega)\|_{l^2}}, \quad \mbox{\boldmath $\psi$}_2 =
\frac{\partial_{\omega} \mbox{\boldmath $\phi$}(\omega)}{
\|\partial_{\omega} \mbox{\boldmath $\phi$}(\omega)\|_{l^2}}
\end{equation}
and require that
\begin{equation}
\label{constraints} \langle {\rm Re}{\bf z}(t),\mbox{\boldmath $\psi$}_1
\rangle =  \langle {\rm Im}{\bf z}(t),\mbox{\boldmath $\psi$}_2 \rangle =
0.
\end{equation}
By Lemma \ref{lemma-bifurcation}, both eigenvectors
$\mbox{\boldmath $\psi$}_1$ and $\mbox{\boldmath $\psi$}_2$ are
locally close to $\mbox{\boldmath $\psi$}_0$, the eigenvector of
$H$ for eigenvalue $\omega_0$, in any norm, e.g.
\begin{equation}
\| \mbox{\boldmath $\psi$}_1 - \mbox{\boldmath $\psi$}_0 \|_{l^2}
+ \| \mbox{\boldmath $\psi$}_2 - \mbox{\boldmath $\psi$}_0
\|_{l^2} \leq C \epsilon,
\end{equation}
for some $C > 0$. Although the vector field of the time evolution
problem (\ref{time-evolution-z}) does not lie in the orthogonal
complement of $\mbox{\boldmath $\psi$}_0$, that is in the
absolutely continuous spectrum of $H$, the difference is small for
small $\epsilon > 0$. We shall prove that the conditions
(\ref{constraints}) define a unique decomposition
(\ref{decomposition}).

\begin{lemma}[Decomposition]
\label{lemma-decomposition} Fix $\epsilon > 0$ and $\delta
> 0$ be sufficiently small. Assume that there exists $T = T(\epsilon,\delta)$ and $C_0 >
0$, such that ${\bf u}(t) \in C^1([0,T],l^2)$ satisfies
\begin{equation}
\label{u-bound}
\| {\bf u}(t) - \mbox{\boldmath $\phi$}(\omega_0 +
\epsilon))\|_{l^2} \leq C_0 \delta \epsilon^{\frac{1}{2p}},
\end{equation}
uniformly on $[0,T]$. There exists a unique choice of
$(\omega,\theta) \in C^1([0,T],\mathbb{R}^2)$ and ${\bf z}(t) \in
C^1([0,T],l^2)$ in the decomposition
(\ref{decomposition}) provided the constraints (\ref{constraints})
are met. Moreover, there exists $C
> 0$ such that
\begin{equation}
\label{theta-omega-bounds}
|\omega(t) - \omega_0 - \epsilon | \leq C \delta \epsilon, \quad | \theta(t)| \leq C
\delta, \quad \| {\bf z}(t) \|_{l^2} \leq C \delta \epsilon^{\frac{1}{2p}},
\end{equation}
uniformly on $[0,T]$.
\end{lemma}

\begin{proof}
We write the decomposition (\ref{decomposition}) in the form
\begin{equation}
\label{z-representation}
{\bf z} = e^{i \theta} \left({\bf u} - \mbox{\boldmath $\phi$}(\omega_0+\epsilon)\right) +
\left( e^{i \theta} \mbox{\boldmath $\phi$}(\omega_0+\epsilon) -
\mbox{\boldmath $\phi$}(\omega) \right).
\end{equation}
First, we show that the constraints (\ref{constraints}) give
unique values of $(\omega,\theta)$ satisfying bounds
(\ref{theta-omega-bounds}) uniformly in $[0,T]$ provided the bound
(\ref{u-bound}) holds. To do so, we rewrite (\ref{constraints})
and (\ref{z-representation}) as a fixed-point problem ${\bf
F}(\omega,\theta) = {\bf 0}$, where ${\bf F} : \mathbb{R}^2
\mapsto \mathbb{R}^2$ is given by
$$
{\bf F}(\omega,\theta) = \left[ \begin{array}{c} \langle {\rm Re}
({\bf u} - \mbox{\boldmath $\phi$}^{(0)}) e^{i \theta},\mbox{\boldmath
$\psi$}_1 \rangle + \langle \mbox{\boldmath $\phi$}^{(0)} \cos \theta
- \mbox{\boldmath $\phi$}(\omega),\mbox{\boldmath $\psi$}_1 \rangle \\
\langle {\rm Im} ({\bf u} - \mbox{\boldmath $\phi$}^{(0)}) e^{i
\theta},\mbox{\boldmath $\psi$}_2 \rangle + \langle
\mbox{\boldmath $\phi$}^{(0)} \sin \theta, \mbox{\boldmath $\psi$}_2
\rangle \end{array} \right],
$$
where $\mbox{\boldmath $\phi$}^{(0)} := \mbox{\boldmath
$\phi$}(\omega_0 + \epsilon)$. We note that ${\bf F}$ is $C^1$ in
$(\theta,\omega)$ thanks to Lemma \ref{lemma-bifurcation}. Direct
computations give the vector field
$$
{\bf F}(\omega_0+\epsilon,0) = \left[ \begin{array}{c} \langle
{\rm Re} ({\bf u} - \mbox{\boldmath $\phi$}^{(0)}),\mbox{\boldmath $\psi$}^{(0)}_1 \rangle \\
\langle {\rm Im} ({\bf u} - \mbox{\boldmath
$\phi$}^{(0)}),\mbox{\boldmath $\psi$}^{(0)}_2 \rangle
\end{array} \right]
$$
and the Jacobian $D {\bf F}(\omega_0+\epsilon,0) = {\bf D}_1 +
{\bf D}_2$ with
\begin{eqnarray*}
{\bf D}_1 & = & \left[ \begin{array}{cc} - \langle
\partial_{\omega} \mbox{\boldmath $\phi$}^{(0)},\mbox{\boldmath
$\psi$}_1^{(0)} \rangle & 0
\\ 0 & \langle \mbox{\boldmath $\phi$}^{(0)},
\mbox{\boldmath $\psi$}^{(0)}_2 \rangle \end{array} \right], \\
{\bf D}_2 & = & \left[ \begin{array}{cc} \langle {\rm Re} ({\bf u}
- \mbox{\boldmath $\phi$}^{(0)}), \partial_{\omega}
\mbox{\boldmath $\psi$}_1^{(0)} \rangle & - \langle {\rm Im} ({\bf
u} - \mbox{\boldmath
$\phi$}^{(0)}), \mbox{\boldmath $\psi$}^{(0)}_1 \rangle \\
\langle {\rm Im} ({\bf u} - \mbox{\boldmath
$\phi$}^{(0)}), \partial_{\omega} \mbox{\boldmath $\psi$}^{(0)}_2 \rangle &
\langle {\rm Re} ({\bf u} - \mbox{\boldmath $\phi$}^{(0)}),\mbox{\boldmath $\psi$}^{(0)}_2 \rangle
\end{array} \right],
\end{eqnarray*}
where $\mbox{\boldmath $\psi$}^{(0)}_{1,2} = \mbox{\boldmath
$\psi$}_{1,2} |_{\omega = \omega_0 + \epsilon}$ and
$\partial_{\omega} \mbox{\boldmath $\psi$}^{(0)}_{1,2} =
\partial_{\omega} \mbox{\boldmath $\psi$}_{1,2} |_{\omega =
\omega_0 + \epsilon}$. Thanks to the bound (\ref{u-bound}) and the
normalization of $\mbox{\boldmath $\psi$}_{1,2}$, there exists an
$(\epsilon,\delta)$-independent constant $C_0 > 0$ such that
$$
\| {\bf F}(\omega_0+\epsilon,0) \| \leq C_0 \delta
\epsilon^{\frac{1}{2p}}.
$$
On the other hand, $D {\bf F}(\omega_0+\epsilon,0)$ is invertible
for small $\epsilon > 0$ since
$$
|({\bf D}_1)_{11}| \geq C_1 \epsilon^{\frac{1}{2p}-1}, \quad |({\bf
D}_1)_{22}| \geq C_2 \epsilon^{\frac{1}{2p}}
$$
and
$$
|({\bf D}_2)_{11}| + |({\bf D}_2)_{21}| \leq C_3 \delta \epsilon^{\frac{1}{2p} - 1}, \quad
|({\bf D}_2)_{12}| + |({\bf D}_2)_{22}| \leq C_4 \delta \epsilon^{\frac{1}{2p}},
$$
for some ($\epsilon$,$\delta$)-independent constants $C_1,C_2,C_3,C_4 > 0$. By the Implicit
Function Theorem, there exists a unique root of ${\bf
F}(\omega,\theta) = {\bf 0}$ near $(\omega_0+\epsilon,0)$ for any
${\bf u}(t)$ satisfying (\ref{u-bound}) such that
$$
|\omega(t) - \omega_0 - \epsilon | \leq C \delta \epsilon, \quad | \theta(t)| \leq C
\delta,
$$
for some $C > 0$. Moreover, if ${\bf u}(t) \in C^1([0,T],l^2)$, then
$(\omega,\theta) \in C^1([0,T],\mathbb{R}^2)$. Finally, existence of a unique
${\bf z}(t)$ and the bound $\| {\bf z}(t) \|_{l^2} \leq C \delta \epsilon^{\frac{1}{2p}}$
follow from the representation (\ref{z-representation}) and the triangle inequality.
\end{proof}

Assuming $(\omega,\theta) \in C^1([0,T],\mathbb{R}^2)$ at least
locally in time and using Lemma \ref{lemma-decomposition}, we
define the time evolution of $(\omega,\theta)$ from the
projections of the time evolution equation
(\ref{time-evolution-z}) with the symplectic orthogonality conditions (\ref{constraints}).
The resulting system is written in the matrix--vector form
\begin{equation} \label{3}
{\bf A}(\omega,{\bf z}) \left[ \begin{array}{cc} \dot{\omega} \\
\dot{\theta} - \omega \end{array} \right] = {\bf f}(\omega,{\bf
z}),
\end{equation}
where
$$
{\bf A}(\omega,{\bf z}) =
\left[ \begin{array}{ccc} \langle \partial_{\omega}
\mbox{\boldmath $\phi$}(\omega),\mbox{\boldmath
$\psi$}_1 \rangle - \langle {\rm Re} {\bf z},\partial_{\omega} \mbox{\boldmath
$\psi$}_1 \rangle & \langle {\rm Im} {\bf z},\mbox{\boldmath
$\psi$}_1 \rangle \\
\langle {\rm Im} {\bf z}, \partial_{\omega} \mbox{\boldmath
$\psi$}_2 \rangle & \langle \mbox{\boldmath $\phi$}(\omega) + {\rm Re} {\bf z},
\mbox{\boldmath $\psi$}_2 \rangle \end{array} \right]
$$
and
$$
{\bf f}(\omega,{\bf z}) =  \left[ \begin{array}{l} \langle {\rm
Im} {\bf N}(\mbox{\boldmath $\phi$}+{\bf z})- {\bf W} {\bf z},
\mbox{\boldmath $\psi$}_1 \rangle \\
\langle {\rm Re} {\bf N}(\mbox{\boldmath $\phi$}+{\bf z}) - {\bf
N}(\mbox{\boldmath $\phi$})-(2p+1) {\bf W} {\bf z},
\mbox{\boldmath $\psi$}_2 \rangle
\end{array} \right].
$$

Using an elementary property for power functions
$$
||a+b|^{2p}(a+b)-|a|^{2p}a|\leq C_p (|a|^{2p}|b|+|b|^{2p+1}),
$$
for some $C_p > 0$, where $a,b \in \mathbb{C}$ are arbitrary, we
bound the vector fields of (\ref{time-evolution-z}) and (\ref{3})
by
\begin{eqnarray}
\label{estimate-N}
\|  {\bf N}(\mbox{\boldmath $\phi$}(\omega)+{\bf z}) - {\bf
N}(\mbox{\boldmath $\phi$}(\omega) \|_{l^2} & \leq & C \left( \|
|\mbox{\boldmath $\phi$}(\omega)|^{2p} |{\bf z}| \|_{l^2} +
\| {\bf z} \|_{l^2}^{2p+1} \right), \\
\label{estimate-f} \| {\bf f}(\omega,{\bf z}) \| & \leq & C
\sum_{j=1}^2 \left( \| |\mbox{\boldmath $\phi$}(\omega)|^{2p-1}
|\mbox{\boldmath $\psi$}_j| |{\bf z}|^2 \|_{l^1} + \|
|\mbox{\boldmath $\psi$}_j| |{\bf z}|^{2p+1} \|_{l^1} \right),
\end{eqnarray}
for some $C > 0$, where the pointwise multiplication of vectors on
$\mathbb{Z}$ is understood in the sense
$$
(|\mbox{\boldmath $\phi$}| |\mbox{\boldmath $\psi$}|)_n = \phi_n
\psi_n.
$$
By Lemmas \ref{lemma-bifurcation} and \ref{lemma-decomposition},
${\bf A}(\omega,{\bf z})$ is invertible for a small ${\bf z} \in
l^2$ and a small $\epsilon \in (0,\epsilon_0)$ so that solutions
of system (\ref{3}) satisfy the estimates
\begin{eqnarray}
\label{33} |\dot{\omega}| & \leq & C \epsilon^{2-\frac{1}{p}}
\left( \| |\mbox{\boldmath $\psi$}_1| |{\bf z}|^2 \|_{l^1} + \|
|\mbox{\boldmath $\psi$}_2| |{\bf z}|^2 \|_{l^1} \right), \\
\label{33a} |\dot{\theta}-\omega| & \leq & C
\epsilon^{1-\frac{1}{p}} \left( \| |\mbox{\boldmath $\psi$}_1|
|{\bf z}|^2 \|_{l^1} + \| |\mbox{\boldmath $\psi$}_2| |{\bf z}|^2
\|_{l^1} \right),
\end{eqnarray}
for some $C > 0$ uniformly in $\| {\bf z} \|_{l^2} \leq C_0
\epsilon^{\frac{1}{2p}}$ for some $C_0 > 0$.

\begin{remark}
{\rm The estimates (\ref{33}) and (\ref{33a}) show that if $\|
{\bf z} \|_{l^2} \leq C \delta \epsilon^{\frac{1}{2p}}$ for some
$C > 0$, then
$$
|\omega(t) - \omega(0)| \leq C \delta^2 \epsilon^2, \quad \left|
\theta(t) - \int_0^t \omega(t') dt' \right| \leq C \delta^2 \epsilon,
$$
uniformly on $[0,T]$ for any fixed $T > 0$. These bounds are smaller than
bounds (\ref{theta-omega-bounds}) of Lemma \ref{lemma-decomposition}. They
become comparable with bounds (\ref{theta-omega-bounds}) for larger
time intervals $[0,T]$, where $T \leq \frac{C_0}{\delta \epsilon}$ for some $C_0 > 0$. Our
main task is to extend these bounds globally to $T =
\infty$.}
\end{remark}

By the theorem on orbital stability in \cite{weinstein}, the
trajectory of the DNLS equation (\ref{dNLS}) originating from a
point in a local neighborhood of the stationary solution
$\mbox{\boldmath $\phi$}(\omega(0))$ remains in a local
neighborhood of the stationary solution $\mbox{\boldmath
$\phi$}(\omega(t))$ for all $t \in \mathbb{R}_+$. By a definition
of orbital stability, for any $\mu_0 > 0$ there exists a $\nu_0 >
0$ such that if $|\omega(0) - \omega_0| \leq \nu_0$ then
$|\omega(t) - \omega_0| \leq \mu_0$ uniformly on $t \in
\mathbb{R}_+$. Therefore, there exists a $\delta(\epsilon)$ for
each $\epsilon \in (0,\epsilon_0)$ such that $T(\epsilon,\delta) =
\infty$ for any $\delta \in (0,\delta(\epsilon))$ in Lemma
\ref{lemma-decomposition}. To prove the main result on asymptotic
stability, we need to show that the trajectory approaches to the
stationary solution $\mbox{\boldmath $\phi$}(\omega_{\infty})$ for
some $\omega_{\infty} \in (\omega_0,\omega_0 + \epsilon_0)$. Our
main result is formulated as follows.

\begin{theorem}[Asymptotic stability in the energy space]
\label{theorem-main} Assume (V1)--(V3), fix $\gamma = +1$ and $p
\geq 3$. Fix $\epsilon
> 0$ and $\delta > 0$ be sufficiently small and assume that
$\theta(0) = 0$, $\omega(0) = \omega_0 + \epsilon$,
and
$$
\| {\bf u}(0) - \mbox{\boldmath $\phi$}(\omega_0 + \epsilon) \|_{l^2} \leq
C_0 \delta \epsilon^{\frac{1}{2p}}
$$
for some $C_0 > 0$. Then,
there exist $\omega_{\infty} \in (\omega_0,\omega_0 +
\epsilon_0)$, $(\omega,\theta) \in
C^1(\mathbb{R}_+,\mathbb{R}^2)$, and a solution ${\bf u}(t) \in X:=
C^1(\mathbb{R}_+,l^2)\cap L^6(\mathbb{R}_+,l^\infty)$ to the DNLS equation (\ref{dNLS})
such that
$$
\lim_{t \to \infty} \omega(t) = \omega_{\infty}, \quad 
\| {\bf u}(t) - e^{-i\theta(t)} \mbox{\boldmath
$\phi$}(\omega(t)) \|_{X} \leq C\de\ve^{1/(2p)}.
$$
\end{theorem}

Theorem \ref{theorem-main} is proved in Section 4. To bound
solutions of the time-evolution problem (\ref{time-evolution-z})
in the space $X$ (intersected with some other spaces of technical
nature), we need some linear estimates, which are described in
Section 3.

\section{Linear estimates}

We need several types of linear estimates, each is designed to
control different nonlinear terms of the vector field of the
evolution equation (\ref{time-evolution-z}). For notational
convenience, we shall use $L^p_t$ and $l^q_n$ to denote $L^p$
space on $t \in [0,T]$ and $l^q$ space on $n \in \mathbb{Z}$,
where $T > 0$ is an arbitrary time including $T = \infty$. The notation
$<n> = (1 + n^2)^{1/2}$ is used for the weights in $l^q_n$ norms.
The constant $C > 0$ is a generic constant, which may change from
one line to another line.

\subsection{Decay  and Strichartz estimates}

Under assumptions (V1)--(V2) on the potential, the following result was proved in
\cite{PS}.

\begin{lemma}[Dispersive decay estimates]
\label{lemma-dispersive} Fix $\sigma > \frac{5}{2}$ and assume
(V1)--(V2). There exists a constant $C > 0$ depending on ${\bf V}$
such that
\begin{eqnarray}
\label{eq:15} \left\| \langle n \rangle^{-\si} e^{-i t
H}P_{a.c.}(H) {\bf f} \right\|_{l^2_n}
& \leq & C (1+t)^{-3/2} \| \langle n \rangle^{\si} {\bf f} \|_{l^2_n},  \\
\label{eq:16} \left\| e^{-i t H}P_{a.c.}(H) {\bf f}
\right\|_{l^\infty_n} & \leq & C (1+t)^{-1/3} \| {\bf f}
\|_{l^1_n},
\end{eqnarray}
for all $t \in \mathbb{R}_+$, where $P_{a.c.}(H)$ is the
projection to the absolutely continuous spectrum of
$H$.
\end{lemma}

\begin{remark}
Unlike the continuous case, the upper bound (\ref{eq:16}) is
non-singular as $t \to 0$ because the discrete case always enjoys
an estimate $\left\| {\bf f} \right\|_{l^\infty_n} \leq  \| {\bf
f} \|_{l^2_n} \leq \| {\bf f} \|_{l^1_n}$.
\end{remark}

Using Lemma \ref{lemma-dispersive} and Theorem 1.2 of Keel-Tao \cite{KT},
the following corollary transfers pointwise decay estimates into Strichartz estimates.

\begin{corollary}[Discrete Strichartz estimates]
\label{corollary-Strichartz} There exists a constant $C > 0$ such that
\begin{eqnarray}
\label{eq:Strichartz1} \left\| e^{-i t H} P_{a.c.}(H) {\bf f}
\right\|_{L^6_t l^{\infty}_n \cap L^{\infty}_t l^2_n} & \leq & C \| {\bf f} \|_{l^2_n},  \\
\label{eq:Strichartz2} \left\| \int_0^t e^{-i (t-s) H} P_{a.c.}(H)
{\bf g}(s) ds \right\|_{L^6_t l^{\infty}_n \cap L^{\infty}_t
l^2_n} & \leq & C \| {\bf g} \|_{L^1_t l^2_n},
\end{eqnarray}
where the norm in $L^p_t l^q_n$ is defined by
$$
\| {\bf f} \|_{L^p_t l^q_n} =  \left( \int_{\mathbb{R}_+} \left(
\| {\bf f}(t) \|_{l^q_n} \right)^p dt \right)^{1/p}.
$$
\end{corollary}

\subsection{Time averaged estimates}

To control the evolution of the varying parameters $(\omega,\theta)$, we derive additional
time averaged estimates. Similar to the continuous case, these estimates are only needed
in one dimension, because the time decay provided by the Strichartz estimates is
insufficient to guarantee time integrability of $\dot{\omega}(t)$ and
$\dot{\theta}(t)-\omega(t)$ bounded from above by the estimates (\ref{33}) and (\ref{33a}).
Without the time
integrability of these quantities, the arguments on the decay of various
norms of ${\bf z}(t)$
satisfying the time evolution problem (\ref{time-evolution-z}) cannot be closed.

\begin{lemma}
\label{le:01} Fix $\sigma > \frac{5}{2}$ and assume (V1) and (V2).
There exists a constant $C > 0$ depending on ${\bf V}$ such that
\begin{eqnarray}
\label{eq:01}
\|<n>^{-3/2} e^{-i t H} P_{a.c.}(H) {\bf f} \|_{l^\infty_n L^2_t} & \leq & C\| {\bf f} \|_{l^2_n} \\
\label{eq:02} \left\|\int_{\mathbb{R}_+} e^{-i t H} P_{a.c.}(H)
{\bf F}(s)dt \right\|_{l^2_n}
& \leq & C\|<n>^{3/2} {\bf F} \|_{l^1_nL^2_t}, \\
\label{eq:033} \left\|<n>^{-\si} \int_0^t e^{-i(t-s)H} P_{a.c.}(H)
{\bf F}(s) ds \right\|_{l^\infty_n  L^2_t} & \leq &
C \|<n>^{\si}  {\bf F} \|_{l^1_n L^2_t} \\
\label{eq:0333} \left\|<n>^{-\si} \int_0^t e^{-i(t-s)H}
P_{a.c.}(H) {\bf F}(s) ds \right\|_{l^\infty_n  L^2_t} & \leq &
C \| {\bf F} \|_{L^1_t l^2_n} \\
\label{eq:03} \left\|\int_0^t e^{-i(t-s)H} P_{a.c.}(H) {\bf F}(s)
ds \right\|_{L^6_tl^\infty_n \cap L^{\infty}_t l^2_n} & \leq & C
\|<n>^3 {\bf F} \|_{L^2_t l^2_n}.
\end{eqnarray}
\end{lemma}

To proceed with the proof, let us set up a few notations. First,
introduce the perturbed  resolvent   $R_V(\la):=(H-\la)^{-1}$ for
$\la \in \mathbb{C} \backslash [0,4]$. We proved in \cite[Theorem
1]{PS} that for any fixed $\omega \in (0,4)$, there exists
$R_V^{\pm}(\omega) = \lim_{\epsilon \downarrow 0} R(\omega \pm i
\epsilon)$ in the norm of $B(\si,-\si)$ for any $\sigma >
\frac{1}{2}$, where $B(\si,-\si)$  denotes the space of bounded
operators from $l^2_{\si}$ to $l^2_{-\si}$.

Next, we recall the Cauchy formula for $e^{i t H}$
\begin{equation}
\label{eq:010}
e^{-i t H} P_{a.c.}(H) = \f{1}{\pi} \int_0^4 e^{-i t \omega} {\rm Im} R_V (\omega) d\omega =
\frac{1}{2\pi i} \int_0^4 e^{-i t \omega} \left[ R^+(\omega) - R^-(\omega) \right] d\omega,
\end{equation}
where the integral is understood in norm $B(\si,-\si)$. We shall parameterize the interval $[0,4]$
by  $\omega = 2 - 2 \cos(\theta)$ for $\theta \in [-\pi,\pi]$.

Let $\chi_0, \chi \in C^{\infty}_0: \; \chi_0 +\chi = 1$ for
all $\theta\in [-\pi, \pi]$, so that
$$
{\rm supp} \chi_0 \subset [-\theta_0,\theta_0] \cup (-\pi, -\pi+\theta_0) \cup (\pi-\theta_0, \pi)
$$
and
$$
{\rm supp} \chi \subset [\theta_0/2,\pi-\theta_0/2]
\cup [-\pi+\theta_0/2,-\theta_0/2],
$$
where
 $0< \theta_0 \leq  \frac{\pi}{4}$.  Note   that the support of $\chi$
stays away  from both $0$ and $\pi$. Following Mizumachi \cite{Miz}, the proof of
Lemma \ref{le:01} relies on the technical lemma.

\begin{lemma}
\label{le:08}
Assume (V1) and (V2). There exists a constant $C > 0$ such that
\begin{eqnarray}
\label{eq:05} & & \sup_{n \in \mathbb{Z}} \|\chi R^{\pm}_V(\omega)
{\bf f} \|_{L^2_{\omega}(0,4)}
\leq C\| {\bf f} \|_{l^2_n}, \\
\label{eq:06} & & \sup_{n \in \mathbb{Z}} \| <n>^{-3/2} \chi_0
R^{\pm}_V(\om) {\bf f}\|_{L^2_\om(0,4)}\leq C\| {\bf f} \|_{l^2_n}.
\end{eqnarray}
\end{lemma}

The proof of Lemma \ref{le:08} is developed in Appendix A. Using Lemma \ref{le:08},
we can now prove Lemma \ref{le:01}.

\begin{proof1}{\em of Lemma \ref{le:01}.} Let us first show \eqref{eq:033},
since it can be deduced from \eqref{eq:15}, although, it can also be viewed
(and proved) as a dual of \eqref{eq:01} as well. Indeed,  \eqref{eq:033} is equivalent to
$$
\|<n>^{-\si} \int_0^t e^{-i (t-s)H} P_{a.c.}(H) <n>^{-\si} {\bf
G}(s) ds \|_{l^\infty_n L^2_t}\leq \| {\bf G} \|_{l^1_n L^2_t}.
$$
By the Krein's theorem, for every Banach space $X$, the elements
of the space  $l^1_n(X)$  are weak limits of linear combinations
of functions in the form $\delta_{n,n_0} x$, where $x\in X$,
$n_0\in \mathbb{Z}$, and $\delta_{n,n_0}$ is Kronecker's symbol.
Thus, to prove the last estimate, we need to check if it holds for
$G_n(s) = \de_{n,n_0} g(s)$, where $g\in L^2_t$. By Minkowski's
inequality, the obvious embedding $l^2\hookrightarrow l^\infty$
and the dispersive decay estimate \eqref{eq:15} for any $\sigma >
\frac{5}{2}$, we have
\begin{eqnarray*}
& & \left\| <n>^{-\si} \int_0^t e^{- i (t-s)H} P_{a.c.}(H)
<n>^{-\si} \de_{n,n_0} g(s) ds \right\|_{l^\infty_n L^2_t} \\
& & \leq C \left\| <n>^{-\si} \int_0^t \left\| e^{- i (t-s)H}
P_{a.c.}(H) <n>^{-\si} \de_{n,n_0}
\right\|_{l^2_n} |g(s)| ds \right\|_{L^2_t} \\
& & \leq C \left\| \int_0^t \frac{|g(s)| ds}{(1 + t-s)^{3/2}} \right\|_{L^2_t}\leq C \|g\|_{L^2_t},
\end{eqnarray*}
where in the last step, we have used Hausdorff-Young's inequality
$L^1*L^2 \hookrightarrow  L^2$.

We show next that \eqref{eq:02}, \eqref{eq:0333}, \eqref{eq:03}
follow from \eqref{eq:01}. Indeed, \eqref{eq:02} is simply a dual
of \eqref{eq:01} and \eqref{eq:02} is hence equivalent to
\eqref{eq:01}. For \eqref{eq:0333}, we apply the so-called
averaging principle, which tells us that to prove \eqref{eq:0333},
it is sufficient to show it for ${\bf F}(t)= \delta(t-t_0) {\bf
f}$, where ${\bf f} \in l^2_n$ and $\delta(t-t_0)$ is Dirac's
delta-function. Therefore, we obtain
\begin{eqnarray*}
\left\|<n>^{- \si} \int_0^t e^{- i (t-s)H}  \delta(s - t_0)
P_{a.c.}(H) {\bf f} ds \right\|_{l^\infty_n L^2_t} & = &
  \|<n>^{-\si}  e^{- i (t-t_0)H} P_{a.c.}(H)  {\bf f} \|_{l^\infty_n L^2_t} \\
  & \leq &
  \|<n>^{-3/2}  e^{- i (t-t_0)H} P_{a.c.}(H)  {\bf f} \|_{l^\infty_n L^2_t} \\
  & \leq & C\| {\bf f} \|_{l^2_n},
\end{eqnarray*}
where in the last step, we have used \eqref{eq:01}.

For \eqref{eq:03}, we argue as follows. Define
\begin{eqnarray*}
T {\bf F}(t) & = & \int_{\mathbb{R}} e^{-i(t-s)H} P_{a.c.}(H) {\bf
F}(s)ds \\ & = & e^{-i t H} P_{a.c.}(H) \left( \int_{\mathbb{R}} e^{- i s
H} P_{a.c.}(H) {\bf F}(s)ds \right) \\ & = & e^{-i t H} P_{a.c.}(H) {\bf
f},
\end{eqnarray*}
where  ${\bf f} = \int_{\mathbb{R}} e^{-i s H} P_{a.c.}(H) {\bf
F}(s) ds$. By an application of the Strichartz estimate
(\ref{eq:Strichartz1}) and subsequently \eqref{eq:02}, we obtain
\begin{eqnarray*}
\| T {\bf F} \|_{L^6_t l^\infty_n \cap L^\infty_t l^2_n} \leq C
\|{\bf f}\|_{l^2_n} \leq  \|<n>^{3/2} {\bf F}\|_{l^1_n L^2_t} \leq
C \|<n>^{3} {\bf F}\|_{l^2_n L^2_t} = 
C \|<n>^{3} {\bf F}\|_{L^2_t l^2_n},
\end{eqnarray*}
where in the last two steps, we have used H\"older's
inequality and the fact that $l^2_n$ and $L^2_t$ commute. 
Now, by the Christ-Kiselev's lemma (e.g. Theorem 1.2 in \cite{KT}), we conclude
that the estimate (\ref{eq:03}) applies to $\int_{0}^t
e^{-i(t-s)H} P_{a.c.}(H) {\bf F}(s) ds$, similar to $T {\bf
F}(t)$. To complete the proof of Lemma \ref{le:08}, it only
remains to  prove \eqref{eq:01}. Let us write
\begin{eqnarray*}
e^{-i t H} P_{a.c.}(H) = \chi e^{-i t H} P_{a.c.}(H) +  \chi_0 e^{-i t H} P_{a.c.}(H)
\end{eqnarray*}
Take a test function ${\bf g}(t)$ such that $\| {\bf g} \|_{l^1_n
L^2_t}=1$ and obtain
\begin{eqnarray*}
\left| \dpr{\chi e^{-i t H} P_{a.c.}(H) {\bf f}}{{\bf g}(t)}_{n,t}
\right| & = & \f{1}{\pi} \left|\int_0^4 \dpr{ \chi {\rm Im}
R_V(\om) {\bf f}}{\int_{\mathbb{R}}
e^{-i t \om} {\bf g}(t)dt}_n d\om \right| \\
& \leq & C
\int_0^4 \dpr{ |\chi R_V(\om) {\bf f}|}{|\hat{{\bf g}}(\om)|}_n d\om \\
& \leq & C \|\chi R^{\pm}_V(\om) {\bf f}\|_{l^{\infty}_n
L^2_{\om}(0,4)} \|\hat{{\bf g}}\|_{l^1_n L^2_\om(0,4)}.
\end{eqnarray*}
By Plancherel's theorem, $\|\hat{{\bf g}}\|_{l^1_n L^2_\om(0,4)}
\leq \|\hat{{\bf g}}\|_{l^1_n L^2_\om(\mathbb{R})} \leq \| {\bf g}
\|_{l^1_n L^2_t}=1$. Using \eqref{eq:05}, we obtain
$$
\left\| \chi e^{-i t H} P_{a.c.}(H) {\bf f} \right\|_{l^\infty_n
L^2_t} = \sup_{ \|{\bf g}\|_{l^1_n L^2_t}=1} \left| \dpr{\chi
e^{-i t H} P_{a.c.}(H) {\bf f}}{{\bf g}(t)}_{n,t} \right| \leq  C
\|{\bf f} \|_{l^2_n}.
$$
Similarly, using \eqref{eq:06} instead of \eqref{eq:05}, one concludes
$$
\left\|<n>^{-3/2} \chi_0 e^{-i t H} P_{a.c.}(H) {\bf f}
\right\|_{l^\infty_n L^2_t} = \sup_{ \|<n>^{3/2} {\bf g}\|_{l^1_n
L^2_t}=1} \left| \dpr{\chi_0 e^{-i t H} P_{a.c.}(H) {\bf f}}{{\bf
g}(t)}_{n,t} \right| \leq  C \|{\bf f}\|_{l^2_n}.
$$
Combining the two estimates, we obtain (\ref{eq:01}).
\end{proof1}

\section{Proof of Theorem \ref{theorem-main}}

Let ${\bf y}(t) = e^{-i \theta(t)} {\bf z}(t)$ and write the time-evolution problem
for ${\bf y}(t)$ in the form
$$
i \dot{\bf y} = H {\bf y} + {\bf g}_1 + {\bf g}_2  + {\bf g}_3,
$$
where
\begin{eqnarray*}
{\bf g}_1 = \left( {\bf N}(\mbox{\boldmath $\phi$} + {\bf y} e^{-i \theta}) -
{\bf N}(\mbox{\boldmath $\phi$}) \right) e^{- i \theta}, \;\;
{\bf g}_2 = -(\dot{\theta} - \omega) \mbox{\boldmath $\phi$} e^{-i \theta}, \;\;
{\bf g}_3 = - i \dot{\omega} \partial_{\omega} \mbox{\boldmath $\phi$}(\omega) e^{-i \theta}.
\end{eqnarray*}
Let $P_0 = \langle \cdot,\mbox{\boldmath $\psi$}_0 \rangle \mbox{\boldmath $\psi$}_0$,
$Q = (I - P_0) \equiv P_{a.c.}(H)$, and decompose the solution ${\bf y}(t)$
into two orthogonal parts
$$
{\bf y}(t) = a(t) \mbox{\boldmath $\psi$}_0 + \mbox{\boldmath $\eta$}(t),
$$
where $\langle \mbox{\boldmath $\psi$}_0, \mbox{\boldmath $\eta$}
\rangle=0$ and $a(t) = \langle {\bf y}(t), \mbox{\boldmath
$\psi$}_0\rangle$. The new coordinates $a(t)$ and $\mbox{\boldmath
$\eta$}(t)$ satisfy the time evolution problem
$$
\left\{ \begin{array}{ccl} i \dot{a} & = & \omega_0 a + \langle {\bf g}, \mbox{\boldmath $\psi$}_0 \rangle, \\
i \dot{\mbox{\boldmath $\eta$}} & = & H \mbox{\boldmath $\eta$} + Q {\bf g} \end{array} \right.
$$
where ${\bf g} = \sum_{j=1}^3 {\bf g}_j$. The time-evolution
problem for $\mbox{\boldmath $\eta$} \equiv
P_{a.c.}(H)\mbox{\boldmath $\eta$}$ can be rewritten in the
integral form as
\begin{equation}
\label{integral} \mbox{\boldmath $\eta$}(t) = e^{-i t H} Q \mbox{\boldmath $\eta$}(0)
- i \int_0^t e^{-i (t-s) H} Q {\bf g}(s) ds,
\end{equation}
Fix $\sigma > \frac{5}{2}$ and introduce the norms
\begin{eqnarray*}
&& M_1 = \| \mbox{\boldmath $\eta$} \|_{L^6_t l^\infty_n}, \quad
M_2 = \| \mbox{\boldmath $\eta$} \|_{L^\infty _t l^2_n}, \quad
M_3 = \| <n>^{-\si} \mbox{\boldmath $\eta$} \|_{l^\infty_n L^{2}_t}, \\
&& M_4 = \| a \|_{L^2_t}, \quad M_5 = \| a \|_{L^{\infty}_t}, \quad
M_6 = \| \omega -\omega(0) \|_{L^{\infty}_t},
\end{eqnarray*}
where the integration in $L^p_t$ is performed on an interval
$[0,T]$ for any $T \in (0,\infty)$. Our goal is to show that
$\dot{\omega}$ and $\dot{\theta} -\omega$ are in $L^1_t$, while
the norms above satisfy an estimate of the form
\begin{equation}
\label{eq:055} \suml_{j=1}^5 M_j \leq C \|{\bf y}(0)\|_{l^2_n} + C
\left( \suml_{j=1}^6 M_j \right)^2
\end{equation}
and
\begin{equation}
\label{eq:055a} M_6 \leq C \epsilon^{2 - \frac{1}{p}} (M_3 +
M_4)^2,
\end{equation}
for some $T$-independent constant $C > 0$ uniformly in
$\suml_{j=1}^6 M_j \leq C \delta \epsilon^{\frac{1}{2 p}}$, where
small positive values of $(\epsilon,\delta)$ are fixed by the
initial conditions $\omega(0) = \omega_0 + \epsilon$ and $\|{\bf
y}(0)\|_{l^2_n} \leq C_0 \delta \epsilon^{\frac{1}{2 p}}$ for some
$C_0 > 0$. The estimate (\ref{eq:055}) and (\ref{eq:055a}) allow 
us to conclude, by elementary continuation arguments, that
$$
\suml_{j=1}^5 M_j \leq C \|{\bf y}(0)\|_{l^2_n} \leq C \delta
\epsilon^{\frac{1}{2 p}}
$$
and $|\omega(t) - \omega_0 - \epsilon| \leq C \delta^2 \epsilon^2$
uniformly on $[0,T]$ for any $T \in (0,\infty)$. By interpolation, $a \in L^6_t$
so that ${\bf z}(t) \in L^6([0,T],l^{\infty}_n)$. Theorem
\ref{theorem-main} then holds for $T = \infty$. In particular,
since $\dot{\omega}(t) \in L^1_t(\mathbb{R}_+)$
and $|\omega(t) - \omega_0 - \epsilon| \leq C \delta^2 \epsilon^2$,
there exists $\omega_{\infty} := \lim_{t \to \infty} \omega(t)$ so that
$\omega_{\infty} \in (\omega_0,\omega_0 + \epsilon_0)$. In addition, 
since ${\bf z}(t) \in L^6(\mathbb{R}_+,l^{\infty}_n)$, then
$$
\lim_{t \to \infty} \| {\bf u}(t) - e^{-i \theta(t)} \mbox{\boldmath $\phi$}(\omega(t)) \|_{l^{\infty}_n} =
\lim_{t \to \infty} \| {\bf z}(t) \|_{l^{\infty}_n} = 0.
$$

{\bf Estimates for $M_6$:} By the estimate \eqref{33},  we have
\begin{eqnarray*}
\int_0^T |\dot{\omega}| dt & \leq & C \epsilon^{2-\frac{1}{p}}
\|<n>^{-2 \si} |{\bf y}|^2\|_{l^\infty_n L^1_t} \left( \| <n>^{2 \si}
\mbox{\boldmath $\psi$}_1 \|_{l^1} +
\| <n>^{2 \si} \mbox{\boldmath $\psi$}_2 \|_{l^1}  \right) \\
& \leq & C \epsilon^{2-\frac{1}{p}} \|<n>^{-\si} {\bf y}
\|_{l^\infty_n L^2_t}^2 \\ & \leq & C \epsilon^{2-\frac{1}{p}}
(M_3+M_4)^2,
\end{eqnarray*}
where we have used the fact that $\mbox{\boldmath $\psi$}_1$ and
$\mbox{\boldmath $\psi$}_2$
decay exponentially as $|n| \to \infty$. As a result, we obtain
$$
M_6 \leq \| \dot{\omega} \|_{L^1_t} \leq C
\epsilon^{2-\frac{1}{p}} (M_3 + M_4)^2.
$$
Similarly, we also obtain that
\begin{eqnarray*}
\int_0^T |\dot{\theta} - \omega| dt \leq C \epsilon^{1-\frac{1}{p}}
(M_3+M_4)^2.
\end{eqnarray*}

{\bf Estimates for $M_4$ and $M_5$:} We use the projection formula
$a = \langle {\bf y}, \mbox{\boldmath $\psi$}_0\rangle$ and recall
the orthogonality relation \eqref{constraints}, so that
$$
\langle {\bf z}, \mbox{\boldmath $\psi$}_0\rangle = \langle {\rm
Re}{\bf z}, \mbox{\boldmath $\psi$}_0-\mbox{\boldmath
$\psi$}_1\rangle + i \langle {\rm Im} {\bf z}, \mbox{\boldmath
$\psi$}_0-\mbox{\boldmath $\psi$}_2\rangle.
$$
By Lemma \ref{lemma-bifurcation} and definitions of
$\mbox{\boldmath $\psi$}_{1,2}$ in
(\ref{eigenvectors-normalized}), we have
$$
\| <n>^{2 \si} (\mbox{\boldmath
$\psi$}_0- \mbox{\boldmath $\psi$}_{1,2})\|_{l^2_n} \leq C |\omega-\omega_0|
$$
for some $C > 0$. Provided $\sigma > \frac{1}{2}$, we obtain
\begin{eqnarray*}
M_4 & = & \|\langle {\bf y}, \mbox{\boldmath
$\psi$}_0\rangle\|_{L^2_t} \leq \| \langle {\rm Re} {\bf z},
\mbox{\boldmath $\psi$}_0-\mbox{\boldmath $\psi$}_1\rangle
\|_{L^2_t} +
\| \langle {\rm Im} {\bf z}, \mbox{\boldmath $\psi$}_0-\mbox{\boldmath $\psi$}_2\rangle \|_{L^2_t}\\
& \leq & \|<n>^{-2 \si} {\bf z}\|_{L^2_t l^2_n} \left( \| <n>^{2 \si}
(\mbox{\boldmath $\psi$}_0-\mbox{\boldmath
$\psi$}_1)\|_{L^{\infty}_t l^2_n} +
 \| <n>^{2 \si} (\mbox{\boldmath $\psi$}_0-\mbox{\boldmath $\psi$}_2)
 \|_{L^{\infty}_t l^2_n} \right) \\
& \leq &  C \|<n>^{-\si} {\bf y} \|_{l^\infty_n L^2_t} \|
\omega -\omega_0 \|_{L^{\infty}_t} \leq C (M_3+M_4) M_6
\end{eqnarray*}
and, similarly,
\begin{eqnarray*}
M_5 & = & \|\langle {\bf y}, \mbox{\boldmath
$\psi$}_0\rangle\|_{L^{\infty}_t} \leq \| \langle {\rm Re} {\bf z}, \mbox{\boldmath
$\psi$}_0-\mbox{\boldmath $\psi$}_1\rangle \|_{L^{\infty}_t} +
\| \langle {\rm Im} {\bf z}, \mbox{\boldmath $\psi$}_0-\mbox{\boldmath $\psi$}_2\rangle \|_{L^{\infty}_t} \\
&\leq &  \|{\bf y}\|_{L^\infty_t l^2_n} \left( \|(\mbox{\boldmath $\psi$}_0-\mbox{\boldmath
$\psi$}_1)\|_{L^{\infty}_t l^2_n} +
 \| (\mbox{\boldmath $\psi$}_0-\mbox{\boldmath $\psi$}_2)
 \|_{L^{\infty}_t l^2_n} \right) \leq
C(M_2 + M_5) M_6.
\end{eqnarray*}

{\bf Estimates for $M_3$:} The free solution in the integral equation
(\ref{integral}) is estimated by \eqref{eq:01} as
\begin{eqnarray*}
\|<n>^{-\si} e^{- i t H} Q \mbox{\boldmath $\eta$}(0)\|_{l^\infty_n L^2_t} \leq
\|<n>^{-3/2} e^{- i t H} Q \mbox{\boldmath $\eta$}(0)\|_{l^\infty_n L^2_t} \leq
C \|\mbox{\boldmath $\eta$}(0) \|_{l^2_n}.
\end{eqnarray*}
Since $\dot{\omega}$ and $\dot{\theta} - \omega$ are $L^1_t$ thanks to the
estimates above, we treat the terms of the integral equation
(\ref{integral}) with ${\bf g}_2$ and ${\bf g}_3$ similarly. By \eqref{eq:0333},
we obtain
\begin{eqnarray*}
\|<n>^{-\si} \int_0^t e^{-i (t-s) H} Q {\bf g}_{2,3}(s) ds\|_{l^\infty_n L^2_t}
& \leq & C \|{\bf g}_{2,3}\|_{L^1_t l^2_n}  \\
& \leq & C \left(\|\dot{\theta}- \omega\|_{L^1_t}
\|\mbox{\boldmath $\phi$}(\omega)\|_{L^{\infty}_t l^2_n}
+ \|\dot{\omega}\|_{L^1_t} \|\partial_{\omega} \mbox{\boldmath $\phi$}(\omega) \|_{L^{\infty}_t l^2_n} \right) \\
& \leq & C \epsilon^{1-\frac{1}{2p}} (M_3 + M_4)^2.
\end{eqnarray*}
On the other hand, using the bound (\ref{estimate-N}) on the vector field ${\bf g}_1$,
we estimate by \eqref{eq:033} and \eqref{eq:0333}
\begin{eqnarray*}
&& \|<n>^{- \si} \int_0^t e^{-i (t-s) H} Q {\bf g}_1(s) ds\|_{l^\infty_n L^2_t} \leq
C(\|<n>^{\si} |\mbox{\boldmath $\phi$}(\omega)|^{2p} |{\bf z}|\|_{l^1_n L^2_t}+
\||{\bf z}|^{2p+1}\|_{L^1_t l^2_n})\\
&& \leq C \left(\|<n>^{-\si} {\bf y}\|_{l^\infty_n L^2_t}
\|<n>^{\si}|\mbox{\boldmath $\phi$}(\omega)|^{2p}\|_{L^{\infty}_t l^1_n} +
\|a \|_{L^{2p+1}_t}^{2p+1}\|\mbox{\boldmath $\psi$}_0\|_{l^{2(2p+1)}_n}^{2p+1} +
\|\mbox{\boldmath $\eta$}\|_{L^{2p+1}_t l^{2(2p+1)}_n}^{2p+1} \right) \\
&& \leq C \left( (M_3+M_4) M_6 + M_4^2 M_5^{2p-1} +
\|\mbox{\boldmath $\eta$}\|_{L^{2p+1}_t l^{2(2p+1)}_n}^{2p+1} \right),
\end{eqnarray*}
where we have
$$
\|a \|_{L^{2p+1}_t}^{2p+1} \leq \| a \|_{L^{\infty}_t}^{2p-1} \| a \|_{L^2_t}^2.
$$
and
$$
\|<n>^{\si}|\mbox{\boldmath $\phi$}(\omega)|^{2p}\|_{l^1_n}
\leq C \| \omega - \omega_0\|_{L^{\infty}_t},
$$
the latter estimate follows from Lemma \ref{lemma-bifurcation}.

To deal with the last term in the estimate, we use the Gagliardo-Nirenberg inequality,
that is, for all $2\leq r,w\leq \infty$ such that $\frac{6}{r} + \frac{2}{w} \leq 1$,
there is a $C > 0$ such that
$$
\|\mbox{\boldmath $\eta$}\|_{L^r_t l^w_n} \leq C \left( \| \mbox{\boldmath $\eta$} \|_{L^6_t l^{\infty}_n}
+ \| \mbox{\boldmath $\eta$} \|_{L^{\infty}_t l^2_n} \right) = C (M_1+M_2).
$$
If $p\geq 3$, then $((2p+1), 2(2p+1))$ is a Strichartz
pair satisfying $\frac{6}{2p + 1} + \frac{2}{2(2p+1)} \leq 1$ and hence,
combining all previous inequalities, we have
\begin{eqnarray*}
M_3 \leq C\left( \|\mbox{\boldmath $\eta$}(0)\|_{l^2_n}+
\epsilon^{1-\frac{1}{p}} (M_3+M_4)^2 + (M_3+M_4) M_6 + M_4^2 M_5^{2p-1} +
(M_1+M_2)^{2p+1}\right),
\end{eqnarray*}
which agrees with the estimate (\ref{eq:055}) for any $p \geq 3$. \\

{\bf Estimates for $M_1$ and $M_2$:} With the help of \eqref{eq:Strichartz1}, the free solution is estimated by
$$
\|e^{- i t H} Q \mbox{\boldmath $\eta$}(0)\|_{L^6_t  l^\infty_n \cap L^\infty_t l^2_n}
\leq C \|\mbox{\boldmath $\eta$}(0)\|_{l^2_n}.
$$
With the help of \eqref{eq:Strichartz2}, the nonlinear terms involving ${\bf g}_{2,3}$ are estimated by
\begin{eqnarray*}
\left\| \int_0^t e^{-i (t-s) H} Q {\bf g}_{2,3}(s) ds \right\|_{L^6_t  l^\infty_n \cap L^\infty_t l^2_n}
& \leq &
C \|{ \bf g}_{2,3}\|_{L^1_t l^2_n} \\ & \leq & C \epsilon^{1-\frac{1}{2p}} (M_3 + M_4)^2.
\end{eqnarray*}
The nonlinear term involving ${\bf g}_1$ is estimated by the sum of two
computations thanks to the bound (\ref{estimate-N}). The first computation is
completed with the help of \eqref{eq:03},
\begin{eqnarray*}
\left\| \int_0^t e^{-i (t-s) H}Q |\mbox{\boldmath $\phi$}(\omega)|^{2p} |{\bf y}|
ds \right\|_{L^6_t l^\infty_n \cap L^\infty_t l^2_n} & \leq &
C \| <n>^3 |\mbox{\boldmath $\phi$}(\omega)|^{2p} |{\bf y}|\|_{L^2_t l^2_n} \\
& \leq &
\| <n>^{3 + \si} |\mbox{\boldmath $\phi$}(\omega)|^{2p}  \|_{L^{\infty}_t l^2_n}
\|<n>^{-\si} {\bf y}\|_{l^\infty_n L^2_t} \\
& \leq & C (M_3+M_4) M_6,
\end{eqnarray*}
whereas the second computation is completed with the help of \eqref{eq:Strichartz2},
\begin{eqnarray*}
\left\| \int_0^t e^{-i (t-s) H}Q |{\bf y}|^{2p+1} ds \right\|_{L^6_t  l^\infty_n \cap L^\infty_t l^2_n}
& \leq & C \||{\bf y}|^{2p+1}\|_{L^1_t l^2_n} \leq C \|{\bf y}\|_{L^{2p+1}_t l^{2(2p+1)}_n}^{2p+1} \\
& \leq & C \left( M_4^2 M_5^{2p-1} +  (M_1 + M_2)^{2p+1} \right),
\end{eqnarray*}
provided $p\geq 3$ holds. We conclude that the estimates for $M_1$ and $M_2$ are the same as the one for
$M_3$.

\section{Numerical results}

We now add some numerical computations which illustrate the asymptotic
stability result of Theorem \ref{theorem-main}. In particular, we shall
obtain numerically the rate, at which the localized perturbations approach to
the asymptotic state of the small discrete soliton. One advantage of numerical
computations is that they are not limited to the case of $p \geq 3$
(which is the realm of our theoretical analysis above), but can
be extended to arbitrary $p \geq 1$. In what follows, we illustrate the
results for $p=1$ (the cubic DNLS), $p = 2$ (the quintic DNLS), and $p = 3$
(the septic DNLS).

Let us consider the single-node external potential with $V_n = - \delta_{n,0}$
for any $n \in \mathbb{Z}$. This potential is known (see Appendix A in \cite{KKK}) to have
only one negative eigenvalue at $\omega_0 < 0$, the continuous spectrum at
$[0,4]$, and no resonances at $0$ and $4$, so it satisfies assumptions (V1)--(V3).
Explicit computations show that the eigenvalue exists at $\omega_0 = 2 - \sqrt{5}$
with the corresponding eigenvector $\psi_{0,n} = e^{-\kappa |n|}$ for any $n \in \mathbb{Z}$, where
$\kappa = {\rm arcsinh}(2^{-1})$. The stationary solutions of 
the nonlinear difference equation (\ref{stationaryDNLS})
exist in a local neighborhood of the ground state of $H = -\Delta + {\bf V}$,
according to Lemma \ref{lemma-bifurcation}. We shall consider numerically the case
$\gamma = -1$, for which the stationary solution bifurcates to the domain $\omega < \omega_0$.
Figure \ref{afig2} illustrates the stationary solutions for $p = 1$ and two different values
of $\omega$, showcasing its increased localization (decreasing
width and increasing amplitude), as $\omega$ deviates from $\omega_0$ towards the negative domain.

\begin{figure}
\begin{center}
\includegraphics[height=7cm]{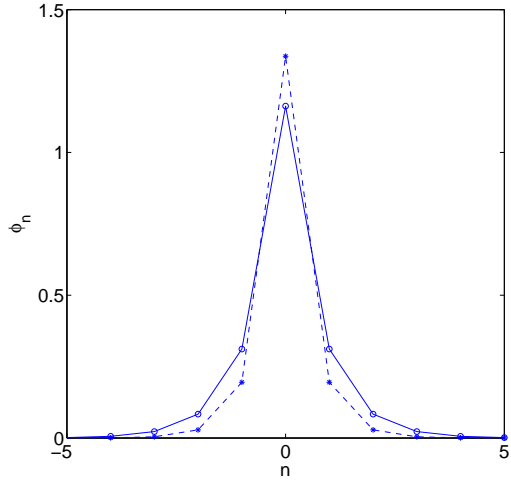}
\end{center}
\caption{Two profiles of the stationary solution of the nonlinear difference equation 
(\ref{stationaryDNLS}) for $V_n = -\delta_{n,0}$, $p = 1$, and for $\omega=-2$ (solid line with circles)
and $\omega=-5$ (dashed line with stars).}
\label{afig2}
\end{figure}

\begin{figure}
\begin{center}
\includegraphics[height=6.8cm]{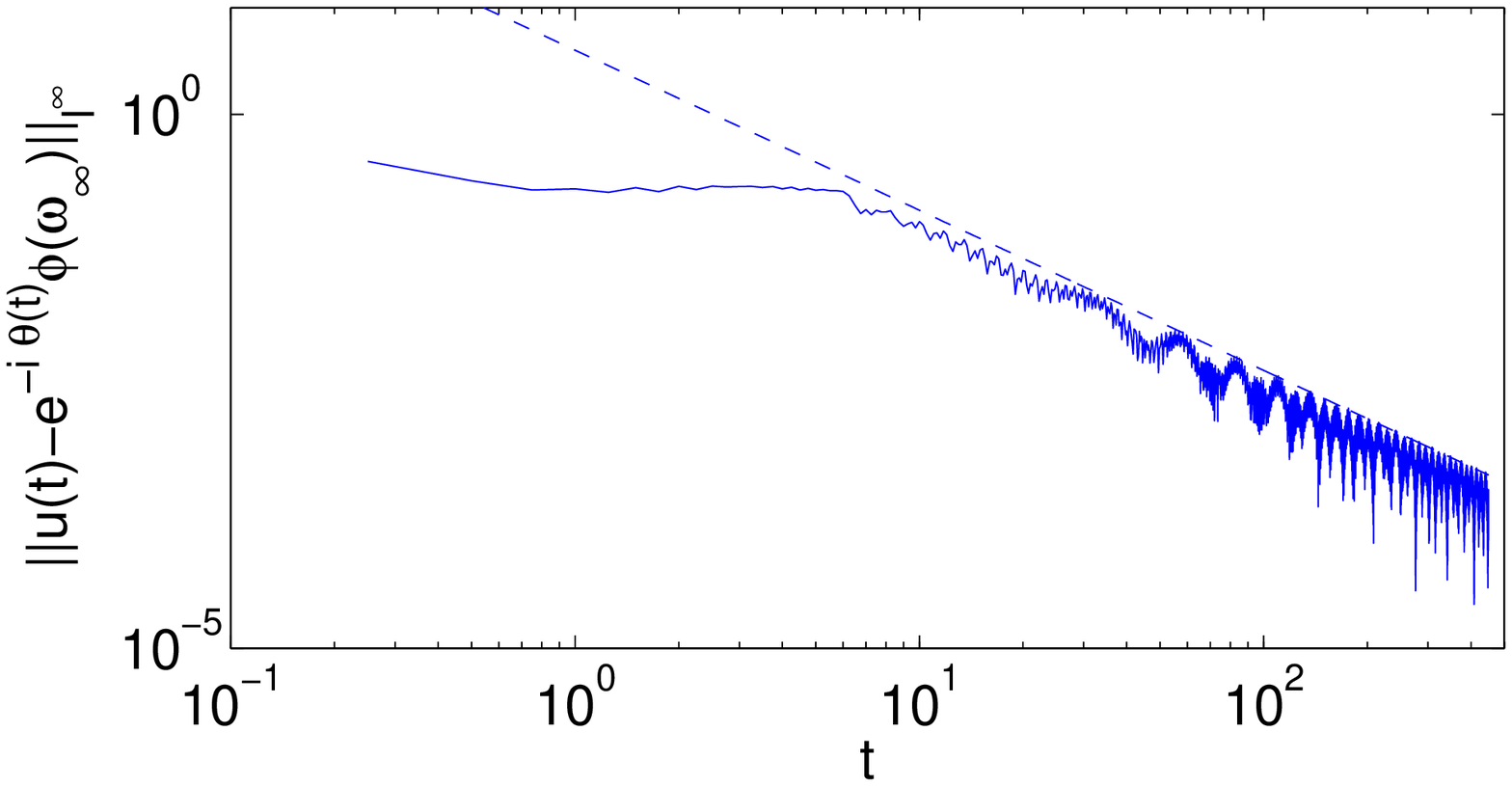}
\includegraphics[height=6.8cm]{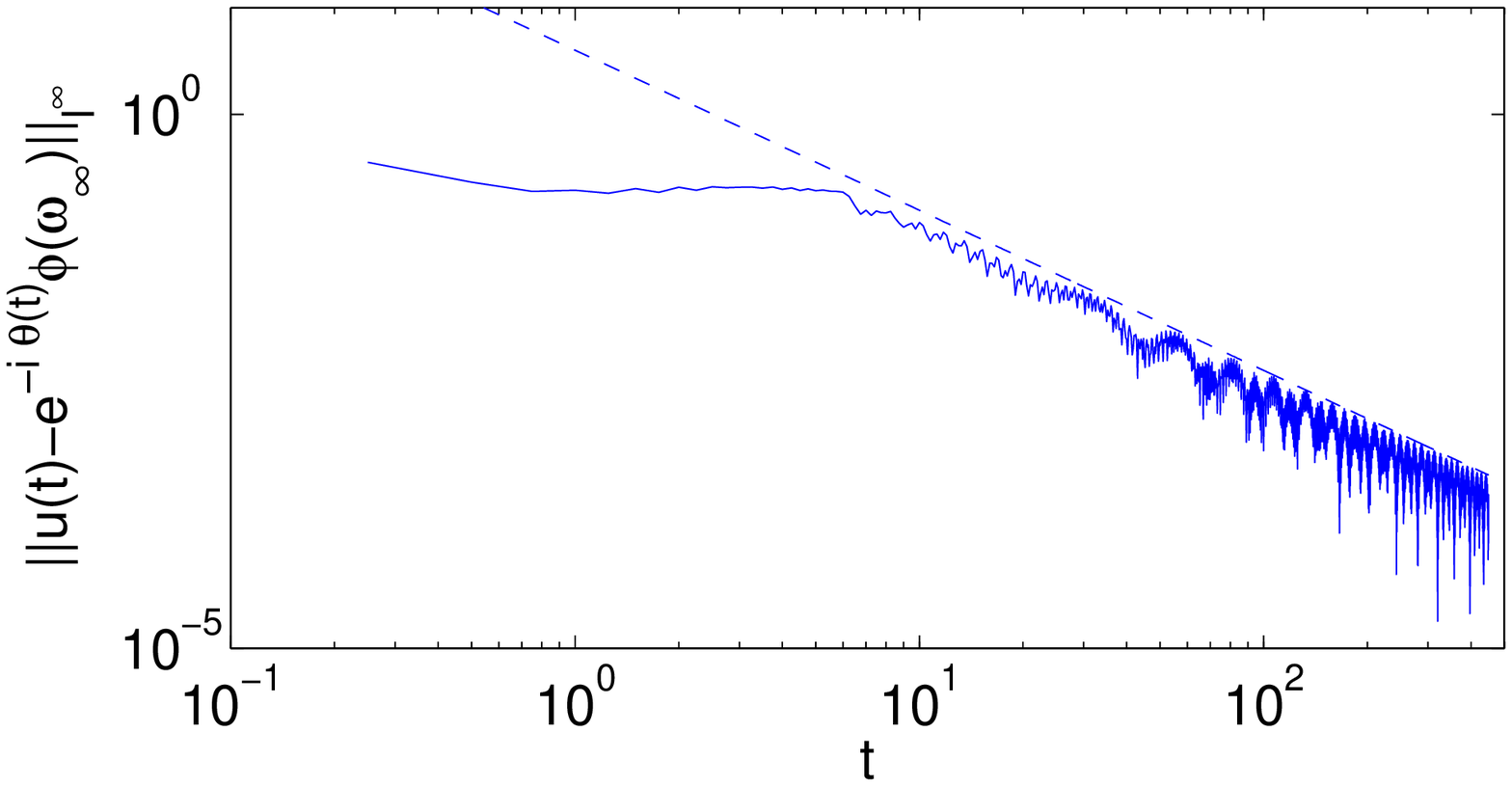}
\includegraphics[height=6.8cm]{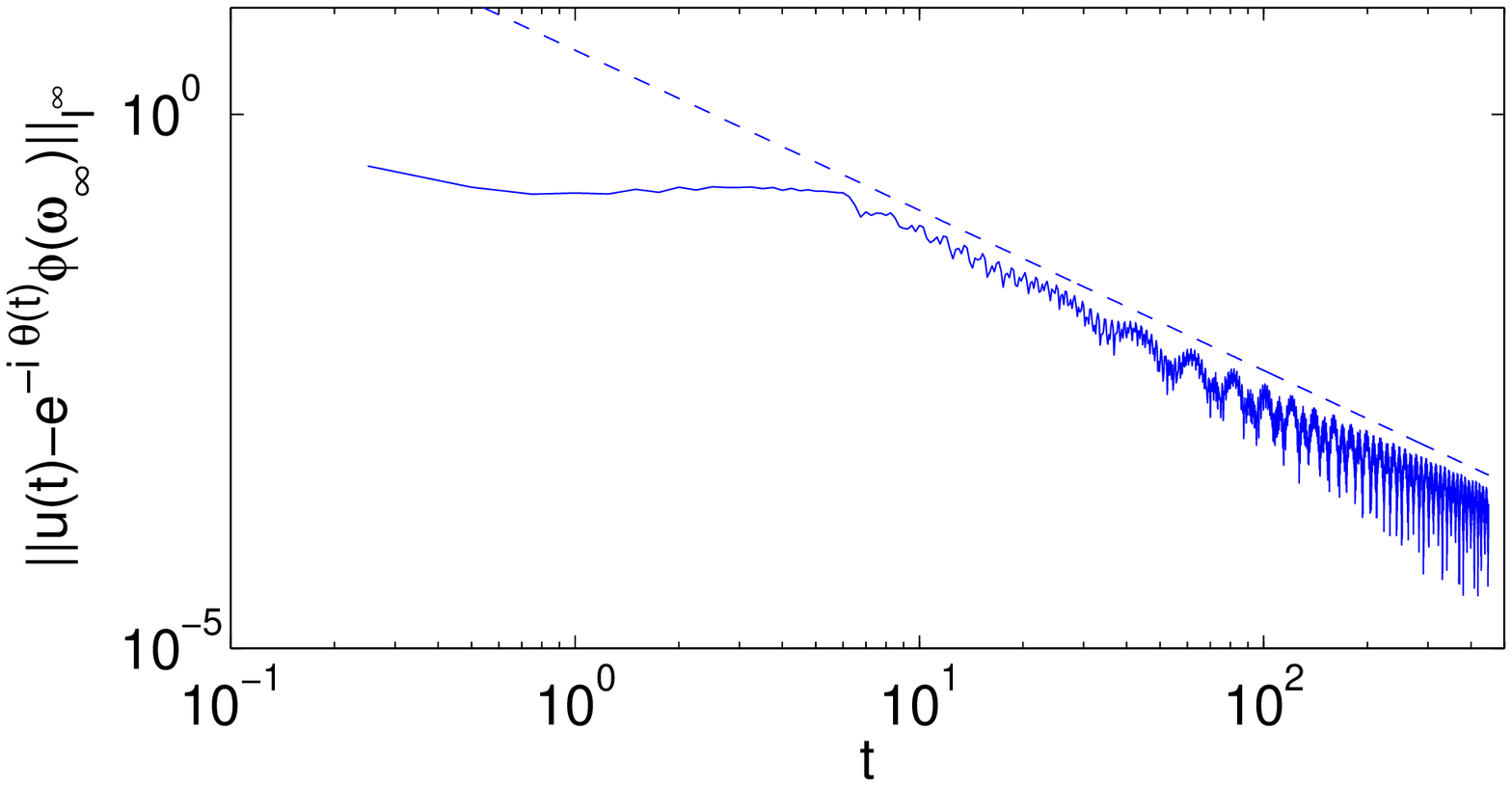}
\end{center}
\caption{Evolution for $p=3$ (top), $2$ (middle), $1$ (bottom)
of $\| {\bf u}(t) - e^{-i \theta(t)} \mbox{\boldmath $\phi$}(\omega_{\infty})\|$ 
as a function of time in a log-log scale (solid)
and comparison with a $t^{-3/2}$ power law decay (dashed) as a guide to
the eye.}
\label{afig4}
\end{figure}

In order to examine the dynamics of the DNLS equation (\ref{dNLS})
we consider single-node initial data $u_n=A \delta_{n,0}$ for any 
$n \in \mathbb{Z}$, with $A=0.75$,
and observe the temporal dynamics of the solution ${\bf u}(t)$. The resulting
dynamics involves the asymptotic relaxation of the localized perturbation into a 
discrete soliton after shedding of some ``radiation''. This dynamics was found to be
typical for
all values of $p = 1,2,3$. In Figure \ref{afig4}, upon suitable subtraction of the phase
dynamics, we illustrate the approach of the wave profile to its asymptotic
form in the $l^{\infty}$ norm. The asymptotic form is obtained by running
the numerical simulation for sufficiently long times, so that the profile
has relaxed to the stationary state. Using a fixed-point algorithm, we identify
the stationary state with the same $l^2$ norm (as the central portion
of the lattice) and confirm that the result of further temporal dynamics is essentially
identical to the stationary state. Subsequently
the displayed $l^{\infty}$ norm of the deviation from the asymptotic
profile is computed, appropriately eliminating the phase by using the gauge invariance of the 
DNLS equation (\ref{dNLS}). 

We have found from Figure \ref{afig4} in the cases $p=3$ (top panel), $p=2$ (middle panel) and
$p=1$ (bottom panel) that the approach to the stationary state follows a power law which is
well approximated as $\propto t^{-3/2}$. The dashed line on all three figures
represents such a decay in each of the cases. We note that the decay rate observed 
in numerical simulations of the DNLS equation (\ref{dNLS}) is faster than the 
decay rate $\propto t^{-1/6-p}$ for any $p > 0$ in Theorem \ref{theorem-main}.

\appendix
\section{Proof of Lemma \ref{le:08}}

For the proof of Lemma \ref{le:08}, we will have to show both
the ``high frequency'' estimate \eqref{eq:05} and
the ``low frequency'' estimate \eqref{eq:06}. To simplify notations, we drop
the bold-face font for vectors on $\mathbb{Z}$ in the appendix.

\subsection{Proof of  \eqref{eq:05}}

Recall the finite  Born series representation of $R_V$
\begin{equation}
\label{eq:012}
R(\omega) = R_0(\omega) - R_0(\omega) V R_0(\omega) +
R_0(\omega) V R(\omega) V R_0(\omega),
\end{equation}
which is basically nothing but the resolvent identity iterated twice.
We have shown in \cite{PS} that for the ``sandwiched resolvent''
$G_{U,W}(\omega) = U R_V(\omega) W$, we have the bounds (see estimate (33) in \cite{PS})
\begin{equation}
\label{eq:011}
\sup_{\theta\in [-\pi, \pi]} \sum_m |G_m(\omega)|+ \left| \f{d}{d\theta} G_m(\omega) \right|
\leq C\|U\|_{\l^2_\si}|W\|_{l^2_\si}.
\end{equation}
for any fixed $\si > \frac{5}{2}$, where $\omega = 2 - 2\cos(\theta)$.

For the three pieces arising from \eqref{eq:012},
similar arguments apply. Starting with the free resolvent term, we have
\begin{eqnarray*}
& &  \sup_{n \in \mathbb{Z}} \int_0^4 \chi |(R_0^{\pm}(\om) f)_n|^2 d\om\leq C
\sup_{n \in \mathbb{Z}} \int_{-\pi}^{\pi} \f{\chi}{\sin(\theta)} \left|\sum_{m \in \mathbb{Z}}
e^{i \theta |m-n|} f_m \right|^2 d\theta \leq \\
& & \leq C \sup_{n \in \mathbb{Z}} \int_{|\theta| \in [\theta_0/2, \pi-\theta_0/2]}
\left( \left| \sum_{m\geq n} e^{i \theta m} f_m \right|^2 + \left|
\sum_{m< n} e^{-i \theta m} f_m \right|^2 \right) d\theta.
\end{eqnarray*}
Introducing the sequence
$$
(g^n)_m:=\left\{\begin{array}{l l} f_m & m\geq n \\ 0 & m<n
\end{array}\right.
$$
we see that the last expression is simply
$C(\|\widehat{g^n}\|_{L^2[\theta_0/2, \pi-\theta_0/2]}^2
+ \|\widehat{f-g^n}\|_{L^2[\theta_0/2, \pi-\theta_0/2]}^2)$, which is equal by Plancherel's identity to
$$
C\|g^n\|_{l^2}^2+\|f-g^n\|_{l^2}^2\leq 2C \|f\|_{l^2}^2.
$$

For the second piece in \eqref{eq:012}, we use that
$\|R_0^{\pm}(\om)\|_{l^1\to l^\infty}\leq C/\sin(\theta)$ and
$|\sin(\theta)| \geq C_0$ on $[\theta_0/2, \pi-\theta_0/2]$
for some $C_0 > 0$, to conlcude
\begin{eqnarray*}
\sup_{n \in \mathbb{Z}} \| \chi R^{\pm}_0(\omega) V R^{\pm}_0(\omega)
f\|_{L^2_\la(0,4)}^2 & \leq & \int_{-\pi}^{\pi}
\f{\chi}{\sin^3(\theta)}  \left(\sum_{n \in \mathbb{Z}} |V_n| |R^{\pm}_0(\omega)
f_n| \right)^2 d\theta \\
& \leq & C \| V \|_{l^1} \sup_{n \in \mathbb{Z}} \int_{-\pi}^{\pi}  \chi
\left| (R^{\pm}_0(\omega) f)_n \right|^2 d\theta,
\end{eqnarray*}
by the triangle inequality. At this point,
we have reduced the estimate to the previous case, provided that $V\in l^1$.

For the third  piece in \eqref{eq:012}, we make use of \eqref{eq:011}.
We have, similar to the previous estimate,
\begin{eqnarray*}
& & \sup_{n \in \mathbb{Z}} \| \chi R^{\pm}_0(\omega) V R^{\pm}_V(\omega) V R^{\pm}_0(\omega)
 f\|_{L^2_\om(0,4)}^2=  \\
& &  \sup_{n \in \mathbb{Z}} \| \chi R^{\pm}_0(\omega) V R_V^{\pm}(\omega) |V|^{1/2} {\rm sgn}(V)
V^{1/2}   R^{\pm}_0(\omega) f\|_{L^2_\om(0,4)}^2= \\
& & \sup_{n \in \mathbb{Z}} \| \chi R^{\pm}_0(\omega) G_{V, |V|^{1/2} {\rm sgn}(V) }[
|V|^{1/2}  R^{\pm}_0(\omega) f]\|_{L^2_\om(0,4)}^2  \\
& & \leq C  \|V\|_{l^2_\si}^2 \||V|^{1/2}\|_{l^2_\si}^2 \||V|^{1/2}\|_{l^1}^2
\sup_{n \in \mathbb{Z}}  \int_{-\pi}^{\pi}  \chi |(R^{\pm}_0(\omega) f)_n|^2 d\theta,
\end{eqnarray*}
where in the  last inequality, we have again reduced the estimate to the first case.

\subsection{Proof of  \eqref{eq:06}}

We only consider the interval $[-\theta_0,\theta_0]$ in the compact support
of $\chi_0(\theta)$ since the arguments for other intervals are similar.
Following the algorithm in \cite{Miz} and the formalism in \cite{PS},
we let $\psi^{\pm}(\theta)$ be two linearly independent solutions of
\begin{equation}
\label{Schrodinger-scattering}
\psi_{n+1} + \psi_{n-1} + (\om - 2) \psi_n = V_n \psi_n, \quad n \in \mathbb{Z},
\end{equation}
according to the boundary conditions $\left| \psi^{\pm}_n - e^{\mp i n \theta} \right| \to 0$
as $n \to \pm \infty$. Let $\psi^{\pm}_n(\theta) = e^{\mp i n \theta} \Psi_n^{\pm}(\theta)$
for all $n \in \mathbb{Z}$. Using the Green function representation, we obtain
\begin{eqnarray*}
\Psi^+_n(\theta) & = & 1 - \frac{i}{2 \sin \theta}
\sum_{m = n}^{\infty} \left( 1 - e^{-2 i \theta (m-n)}
\right) V_m \Psi^+_m(\theta),
\\ \Psi^-_n(\theta) & = & 1 - \frac{i}{2 \sin \theta}
\sum_{m = -\infty}^{n} \left( 1 - e^{-2i \theta (n-m)}
\right) V_m \Psi^-_m(\theta).
\end{eqnarray*}
The discrete Green
function for the resolvent operators $R^{\pm}(\omega)$ has the kernel
$$
\left[ R_V^{\pm}(\omega) \right]_{n,m} = \frac{1}{W(\theta_{\pm})} \left\{ \begin{array}{cc}
\psi_n^+(\theta_{\pm}) \psi_m^-(\theta_{\pm}) \;\; \mbox{for} \;\; n \geq m \\
\psi_m^+(\theta_{\pm}) \psi_n^-(\theta_{\pm}) \;\; \mbox{for} \;\; n < m \end{array} \right.
$$
where $\theta_- = -\theta_+$, $\theta_- \in [0,\pi]$ for $\omega \in [0,4]$, and
$W(\theta) = W[\psi^+,\psi^-] = \psi_n^+ \psi^-_{n+1} - \psi^+_{n+1} \psi^-_n$
is the discrete Wronskian, which is independent of $n \in \mathbb{Z}$.
We need to estimate
$$
\| \chi_0  R_V^{\pm}(\omega) f \|^2_{L^2_{\omega}(0,4)} =
\int_{-\pi}^{\pi} \frac{2 \chi^2_0  \sin \theta d \theta}{W^2(\omega)}
\left( \sum_{m = -\infty}^{n-1} \psi_n^+(\theta) \psi_m^-(\theta) f_m +
\sum_{m = n}^{\infty} \psi_n^-(\theta) \psi_m^+(\theta) f_m \right)^2.
$$
We may assume that $n \geq 1$ for definiteness and split
$$
\sum_{m = -\infty}^{n-1} \psi_m^-(\theta) f_m = \sum_{m=0}^{n-1} \psi_m^-(\theta) f_m +
\sum_{m=-\infty}^{-1} e^{i m \theta} f_m + \sum_{m=-\infty}^{-1} e^{i m \theta} (\Psi_m^- - 1) f_m :=
I_1 + I_2 + I_3
$$
and
$$
\sum_{m = n}^{\infty} \psi_m^+(\theta) f_m =
\sum_{m = n}^{\infty} e^{-i m \theta} f_m + \sum_{m = n}^{\infty}
e^{-i m \theta} \left( \Psi_m^+(\theta) - 1 \right) f_m := I_4 + I_5
$$
We are using the scattering theory from \cite{PS} to claim that
\begin{equation}
\label{properties-limiting-functions}
\sup_{\theta \in [-\theta_0,\theta_0]} \left( \| \Psi^{\pm}(\theta) \|_{l^{\infty}(\mathbb{Z}_{\pm})}
+  \| \langle n \rangle^{-1} \Psi^{\pm}(\theta) \|_{l^{\infty}(\mathbb{Z}_{\mp})} \right) < \infty,
\end{equation}
where $\langle n \rangle = (1 + n^2)^{1/2}$. Then, we have
\begin{eqnarray*}
| I_1 | & \leq & \left( \sum_{m=0}^{n-1} |\Psi_m^-(\theta)|^2 \right)^{1/2}
\left( \sum_{m=0}^{n-1} |f_m|^2 \right)^{1/2} \leq C_1 \langle n \rangle^{3/2}
\| f \|_{l^2}, \\
| I_3 | & \leq & \left( \sum_{m=-\infty}^{-1} |\Psi_m^-(\theta) - 1|^2 \right)^{1/2}
\left( \sum_{m=-\infty}^{-1} |f_m|^2 \right)^{1/2} \leq C_3
\left\| \sum_{k =-\infty}^m |m-k| |V_k| \right\|_{l^2_m(\mathbb{Z}_-)}
\| f \|_{l^2},\\
| I_5 | & \leq & \left( \sum_{m=n}^{\infty} |\Psi_m^+(\theta) - 1|^2 \right)^{1/2}
\left( \sum_{m=n}^{\infty} |f_m|^2 \right)^{1/2} \leq C_5
\left\| \sum_{l = m}^{\infty} |k-m| |V_k| \right\|_{l^2_m(\mathbb{Z}_+)}
\| f \|_{l^2},
\end{eqnarray*}
for some $C_1,C_3,C_5 > 0$. We note that
$$
\left\| \sum_{k =-\infty}^m |m-k| |V_k| \right\|_{l^2_m(\mathbb{Z}_-)} \leq
\left\| \sum_{k =-\infty}^m |m-k| |V_k| \right\|_{l^1_m(\mathbb{Z}_-)} \leq C_4
\| V \|_{l^1_2},
$$
for some $C_4 > 0$. Therefore, the brackets in $I_3$ and $I_5$ are bounded
if $V \in l^1_{2\si}$ for $\si > \frac{5}{2}$.
Since $I_2$ and $I_4$ are given by the discrete Fourier
transform, Parseval's equality implies that
$$
\int_{-\pi}^{\pi} \left( I^2_2 + I^2_4 \right) d \theta \leq C_2 \| f \|_{l^2}^2,
$$
for some $C_2 > 0$. Using now the fact that $|W(\theta)| \geq W_0$ and
$|\sin \theta| \leq C_0$ uniformly in $[-\theta_0,\theta_0]$, the support of
$\chi_0(\theta)$, and using the property (\ref{properties-limiting-functions}), we obtain
$$
\| \chi_0  R_V^{\pm}(\omega) f \|^2_{L^2_{\omega}(0,4)} \leq
C \left( 1 + \langle n \rangle^2 +  \langle n \rangle^3 \right) \| f \|^2_{l^2},
$$
which gives (\ref{eq:06}).


\begin{thebibliography}{99}
\bibitem{BP1} Buslaev V.S; Perelman G.S. ``Scattering for the nonlinear Schr\"{o}dinger
equation: states close to a soliton'', St. Petersburg Math. J. {\bf 4} (1993), 1111--1142.

\bibitem{BP2} Buslaev V.S; Perelman G.S. ``On the stability of solitary waves for nonlinear
Schr\"{o}dinger equations'', Amer. Math. Soc. Transl. {\bf 164} (1995), 75--98.

\bibitem{BS} Buslaev V.S; Sulem C. ``On asymptotic stability of solitary waves for
nonlinear Schr\"{o}dinger equations'', Ann. Inst. H. Poincar\'{e} Anal. Non Lineare
{\bf 20} (2003), 419--475.

\bibitem{cuccagna} Cuccagna, S. ``A survey on asymptotic stability of ground states of
nonlinear Schr\"{o}dinger equations'' in \emph{Dispersive nonlinear problems in mathematical physics},
pp. 21--57 (Quad. Mat., 15, Dept. Math., Seconda Univ. Napoli, Caserta, 2004)

\bibitem{Cuc} Cuccagna S. ``On asymptotic stability in energy space of ground states of NLS in 1D",
preprint, arXiv:0711.4192v2  (November, 2007)

\bibitem{CT} Cuccagna S.; Tarulli, M. ``On asymptotic stability of standing waves
of discrete Schr\"{o}dinger equation in $\mathbb{Z}$",
preprint, arXiv:0808.2024v1  (August, 2008)

\bibitem{mora} H.S. Eisenberg, Y. Silberberg, R. Morandotti,
A.R. Boyd and J.S. Aitchison,
``Discrete spatial optical solitons in waveguide arrays'',
Phys. Rev. Lett. {\bf 81}, 3383-3386 (1998).


\bibitem{GS1} Gang Z; Sigal, I.M. ``Asymptotic stability of nonlinear Schr\"{o}dinger equations
with potential'', Rev. Math. Phys. {\bf 17} (2005), 1143--1207.

\bibitem{GS2} Gang Z; Sigal, I.M. ``Relaxation of solitons in nonlinear Schr\"{o}dinger equations
with potential'', Adv. Math. 216 (2007), 443--490.

\bibitem{KT}
M. Keel, T. Tao, \emph{Endpoint Strichartz estimates},
{\em Amer. J. Math.}, {\bf 120} (1998), 955--980.

\bibitem{KEDS} Kevrekidis, P.G.; Espinola--Rocha, J.A.; Drossinos, Y.; Stefanov, A. ``Dynamical
barrier for the formation of solitary waves in discrete lattices'', Phys. Lett. A {\bf 372}
(2008), 2237--2253.



\bibitem{KKK}
Komech, A; Kopylova, E.; Kunze, M. ``Dispersive estimates
 for 1D discrete Schr\"odinger and Klein-Gordon equations'', {\em  Appl. Anal.} {\bf  85}
(2006), 1487--1508.

\bibitem{rosberg} M. Matuszewski, C.R. Rosberg, D.N. Neshev,
A.A. Sukhorukov, A. Mitchell, M. Trippenbach, M.W. Austin,
W. Krolikowski and Yu.S. Kivshar,
``Crossover from self-defocusing to discrete trapping in nonlinear waveguide
arrays'',
Opt. Express {\bf 14}, 254-259 (2006).


\bibitem{Nirenberg} Nirenberg, L. {\em Topics in nonlinear
functional analysis}, Courant Lecture Notes in Mathematics {\bf 6}
(AMS, New York, 2001).

\bibitem{Miz} Mizumachi T. ``Asymptotic stability of small solitons to 1D NLS with potential",
preprint, arXiv:math/0605031v2 (May, 2008).

\bibitem{kroli} F. Palmero, R. Carretero-Gonz{\'a}lez, J. Cuevas,
P.G. Kevrekidis and W. Kr{\'o}likowski,
``Solitons in one-dimensional nonlinear Schr{\"o}dinger lattices
with a local inhomogeneity'', Phys. Rev. E {\bf 77}, 036614 (2008).

\bibitem{PP} Panayotaros, P.; Pelinovsky, D; ``Periodic oscillations of discrete NLS solitons
in the presence of diffraction management'', Nonlinearity {\bf 21}  (2008), 1265--1279.

\bibitem{PS} Pelinovsky, D; Stefanov, A. ``On the spectral theory and dispersive estimates
for a discrete Schr\"{o}dinger equation in one dimension'', J. Math. Phys., to be printed
(November, 2008).

\bibitem{PW} Pillet, C.A.; Wayne, C.E. ``Invariant manifolds for a class of dispersive, Hamiltonian,
partial differential equations'', J. Diff. Eqs. {\bf 141} (1997), 310--326.


\bibitem{SW1} Soffer, A; Weinstein, M.I. ``Multichannel nonlinear scattering theory for nonintegrable
equations'', Comm. Math. Phys. {\bf 133} (1990), 119--146.

\bibitem{SW2} Soffer, A; Weinstein, M.I. ``Multichannel nonlinear scattering theory for nonintegrable
equations II: The case of anisotropic potentials and data'', J. Diff. Eqs. {\bf 98} (1992), 376--390.

\bibitem{SW3} Soffer, A; Weinstein, M.I. ``Selection of the ground state for nonlinear
Schr\"{o}dinger equations'', Rev. Math. Phys. {\bf 16} (2004), 977--1071.

\bibitem{SK}
Stefanov, A.; Kevrekidis, P. ``Asymptotic behaviour of small
solutions for the discrete nonlinear Schr\"odinger and Klein-Gordon
equations'', {\em Nonlinearity} {\bf  18} (2005), 1841--1857.

\bibitem{weinstein}
Weinstein, M. I. ``Excitation thresholds for nonlinear localized modes on lattices'',
\emph{Nonlinearity} 12 (1999), 673--691.

\bibitem{YT1} Yau, H.T.; Tsai, T.P. ``Asymptotic dynamics of nonlinear Schr\"{o}dinger
equations: resonance dominated and radiation dominated solutions'', Comm. Pure Appl. Math.
{\bf 55} (2002), 1--64.

\bibitem{YT2} Yau, H.T.; Tsai, T.P. ``Stable directions for excited states of nonlinear
Schr\"{o}dinger equations'', Comm. Part. Diff. Eqs. {\bf 27} (2002), 2363--2402.

\bibitem{YT3} Yau, H.T.; Tsai, T.P. ``Relaxation of excited states in nonlinear Schr\"{o}dinger
equations'', Int. Math. Res. Not. (2002), 1629--1673.

\end{thebibliography}
\end{document}